\documentclass[preprint]{imsart}

\RequirePackage[OT1]{fontenc}
\RequirePackage{amsthm,amsmath,amssymb}
\RequirePackage[numbers]{natbib}
\RequirePackage[colorlinks,citecolor=blue,urlcolor=blue]{hyperref}

\usepackage[utf8]{inputenc}
\usepackage{bbm}
\usepackage{comment}
\usepackage{amsmath}
\usepackage{amsfonts}
\usepackage{amssymb}
\usepackage{graphicx}
\usepackage{mathtools}
\usepackage{bm}
\usepackage{dsfont}
\usepackage{booktabs}
\usepackage{subfigure}
\usepackage[a4paper,margin=2.65cm]{geometry}

\newtheorem{theorem}{Theorem}

\newtheorem{corollary}[theorem]{Corollary}
\newtheorem{definition}[theorem]{Definition}
\newtheorem{example}[theorem]{Example}
\newtheorem{lemma}[theorem]{Lemma}
\newtheorem{proposition}[theorem]{Proposition}

\startlocaldefs

\newcommand{\1}[1]{\mathds{1}_{#1}}
\newcommand{\doverline}[1]{\overline{\overline{#1}}}
\numberwithin{equation}{section}
\theoremstyle{plain}

\newcommand{\var}{\mathrm{var}}
\DeclareMathOperator{\CRM}{CRM}
\DeclareMathOperator{\GGP}{GGP}

\DeclareMathOperator{\Poi}{Poisson}

\DeclareMathOperator{\Gam}{Gamma}

\endlocaldefs

\begin{document}

\begin{frontmatter}
\title{Non-exchangeable random partition models for microclustering}
\runtitle{Non-exchangeable random partition models}

\begin{aug}
\author{\fnms{Giuseppe} \snm{Di Benedetto}\thanksref{m1}\ead[label=e1]{giuseppe.dibenedetto@spc.ox.ac.uk}},
\author{\fnms{Fran\c cois} \snm{Caron}\thanksref{m1}\ead[label=e2]{caron@stats.ox.ac.uk}}
\and
\author{\fnms{Yee Whye} \snm{Teh}\thanksref{m1,m2}\ead[label=e3]{y.w.teh@stats.ox.ac.uk}
\ead[label=u1]{http://www.foo.com}}

\runauthor{G. di Benedetto et al.}

\address{University of Oxford\thanksmark{m1}\and Google Deepmind\thanksmark{m2}}


\end{aug}

\begin{abstract}
Many popular random partition models, such as the Chinese restaurant process and its two-parameter extension, fall in the class of exchangeable random partitions, and have found wide applicability in model-based clustering, population genetics, ecology or network analysis. While the exchangeability assumption is sensible in many cases, it has some strong implications. In particular, Kingman's representation theorem implies that the size of the clusters necessarily grows linearly with the sample size; this feature may be undesirable for some applications, as recently pointed out by Miller et al. (2015). We present here a flexible class of  non-exchangeable random partition models which are able to generate partitions whose cluster sizes grow sublinearly with the sample size, and where the growth rate is controlled by one parameter. Along with this result, we provide the asymptotic behaviour of the number of clusters of a given size, and show that the model can exhibit a power-law behavior, controlled by another parameter. The construction is based on completely random measures and a Poisson embedding of the random partition, and inference is performed using a Sequential Monte Carlo algorithm. Additionally, we show how the model can also be directly used to generate sparse multigraphs with power-law degree distributions and degree sequences with sublinear growth. Finally, experiments on real datasets emphasize the usefulness of the approach compared to a two-parameter Chinese restaurant process.
\end{abstract}

\end{frontmatter}

\section{Introduction}

Random partitions arise in a wide range of different applications such as Bayesian model-based clustering~\citep{Lau2007,Mueller2013}, population genetics~\citep{Kingman1978}, ecology~\citep{Lijoi2007a} or network modelling~\citep{Bloem-Reddy2017}. A partition of a set $[n]=\{1,\dots,n\}$ is a set of disjoint non-empty subsets $A_{n,j}\subseteq  [n]$, $j=1,\ldots,K_n$ with $\cup_j A_{n,j}=[n]$ where $K_n\leq n$ is the number of clusters and  $A_{n,j}$ denotes the set of integers in cluster $j$. A random partition $\Pi_n$ of $[n]$ is a random variable taking values in the finite set of partitions of $[n]$. A random partition of $\mathbb{N}$ is a sequence $\Pi=(\Pi_n)_{n\geq 1}$ of random partitions of $[n]$, defined on a common probability space, that satisfy the Kolmogorov consistency condition: for every $1\le m<n$, $\Pi_n$ restricted to $[m]$ is $\Pi_m$~\citep{Kingman1975,Aldous1985,Pitman2002,Pitman1995}. For many applications, it is important to characterize the properties of the random partition model as the number of items $n$ grows. Of particular importance are the asymptotic behavior of (i) the number of clusters, (ii) the proportion of clusters of a given size, and (iii) the cluster sizes.\smallskip

In some contexts a natural and useful assumption is the \emph{exchangeability} of the random partition: $\Pi$ is said to be exchangeable if for every $n\ge1$ the distribution of $\Pi_n$ is invariant to the group of permutations of $[n]$. Arguably the best known exchangeable random partition model is the Chinese Restaurant Process (CRP)~\cite{Aldous1985}. This model has a single parameter, a very simple generative process and well established asymptotic properties; the number of clusters $K_n$ grows logarithmically with $n$~\cite{Korwar1973}, while the proportion of clusters of any given size goes to zero. Such behaviour is not appropriate for some applications such as natural language processing or image segmentation~\citep{Teh2006,Sudderth2009}, where these proportions typically exhibit a power-law behavior. This asymptotic property can be achieved by considering the two-parameter CRP~\cite{pitman1997}, another exchangeable random partition model which generalizes the one-parameter CRP. Beyond these two popular models, the class of exchangeable random partitions offers a rich, flexible and tractable framework, including models based on normalized random measures~\citep{Regazzini2003,James2002,Lijoi2007,James2009}, Poisson-Kingman processes~\citep{Pitman2003} or Gibbs-type priors~\citep{Gnedin2006,Favaro2013,DeBlasi2015,Bacallado2015}.\smallskip

Although exchangeability is a sensible assumption in many applications, it has strong implications regarding the growth rate of the cluster's sizes: Kingman's representation theorem indeed implies that the size of each cluster grows linearly with the sample size $n$. As recently noted by Miller et al.~\cite{Miller2015} and Betancourt et al.~\cite{Betancourt2016}, this assumption may be unrealistic for some applications, such as entity resolution, which require the construction of random partition models where the cluster sizes grow sublinearly with the sample size; Miller et al.~\cite{Miller2015} call it the \emph{microclustering} property. \smallskip

The objective of this article is to present a general class of models for non-exchangeable random partitions of $\mathbb N$ which retains the wide range of asymptotic properties of exchangeable partition models, while capturing the microclustering property. The model allows:
\begin{itemize}
\item Flexibility in the asymptotic growth rates of (i) the number of cluster, and (ii) the proportion of clusters of a given size, tuned by interpretable parameters; in particular, it is possible to obtain the same growth rates as with the two-parameter CRP, including the power-law regime.
\item Flexibility in the asymptotic sublinear growth rates of the cluster sizes, tuned by interpretable parameters.
\end{itemize}

The paper is organized as follows. In Section~\ref{sec:preliminaries} we provide background on completely random measures (CRM), exchangeable random partitions, and give a derivation of the partition associated to a normalized completely random measure via a Poissonization technique. In Section~\ref{sec:model} we present our novel class of non-exchangeable random partition models that builds on the same Poissonization idea. Section~\ref{sec:properties} develops the properties of this class of models and posterior inference. In Section~\ref{sec:graphs}, we describe how our model can also be used to build sparse random multigraph models with an asymptotic power-law degree distribution and sublinear degree growth. Section~\ref{sec:discussion} discusses related approaches in the literature. Section~\ref{experiments} provides comparisons between the proposed non-exchangeable model and the two-parameter CRP on two datasets. Most proofs and some definitions can be found in the Appendix.

\section{Background material}\label{sec:preliminaries}

\subsection{Completely random measures}
Completely random measures, introduced by Kingman~\citep{Kingman1967}, have found wide applicability as priors over functional spaces in Bayesian nonparametrics~\citep{Regazzini2003,Lijoi2010}, due to their flexibility and tractability; the reader can refer to \citep[Chapter 10.1]{Daley2008} or \citep{Lijoi2010} for an extended coverage. A homogeneous CRM on $\mathbb R_+$ without fixed atoms nor deterministic component is almost surely discrete and takes the form
$$
W = \sum_{j\ge1}\omega_j \,\delta_{\vartheta_j}
$$
where $\{(\omega_j,\vartheta_j)\}_{j\geq 1}$ are the points of a Poisson process on $(0,\infty)\times \mathbb R_+$ with mean measure $\nu(d\omega,d\theta)$. The measure decomposes as $\nu(d\omega,d\theta)=\rho(d\omega)\alpha(d\theta)$ where $\alpha$ is a non-atomic Borel measure on $\mathbb R_+$, called the base measure, such that $\alpha(A)<\infty$ for any bounded Borel set $A$, and $\rho$ is a L\'evy measure on $(0,\infty)$. We write $W\sim \CRM(\alpha,\rho)$. We will also assume in the following that the base measure $\alpha(d\theta)$ is absolutely continuous with respect to the Lebesgue measure and
\begin{equation}
\int_{(0,\infty)\times\mathbb R_+} \rho(dw)\alpha(d\theta)=\infty.
\label{eq:conditionmeasure}
\end{equation}

Let
\begin{equation}
\psi(t)=\int_0^\infty \left \{1-e^{-wt} \right \}\rho(dw)
\end{equation}
be the Laplace exponent and define, for any integer $m\geq 1$ and any $u>0$
\begin{align*}
\kappa(m,u)=\int_0^\infty\omega^me^{-u\omega} \rho(d\omega).
\end{align*}

A remarkable example of CRM is the generalized gamma process \citep{Hougaard1986,Brix1999} (GGP) with mean measure
\[
\nu(d\omega,d\theta) = \frac{1}{\Gamma(1-\sigma_0)}\omega^{-1-\sigma_0}e^{-\zeta_0\omega}d\omega\,\alpha(d\theta)
\]
with $\sigma_0\in(0,1)$ and $\zeta_0\geq 0$ or $\sigma_0\in(-\infty,0]$ and $\zeta_0>0$. We write $W\sim \GGP(\alpha,\sigma_0,\zeta_0)$. The GGP has been a popular model in Bayesian nonparametrics due to its flexibility and attractive conjugacy properties~\citep{James2002,Lijoi2003,Lijoi2007,Caron2017}. It includes several important models as special cases: the gamma process for $\sigma_0=0$, $\zeta_0>0$; the inverse gaussian process for $\sigma_0=1/2$, $\zeta_0>0$ and the stable process for $\sigma_0\in(0,1)$ and $\zeta_0=0$.

\subsection{Exchangeable random partitions}
For an exchangeable partition $\Pi=(\Pi_n)_{n\geq 1}$ of $\mathbb N$ we have, for every $n\geq 1$
\[
\Pr(\Pi_n = \{A_{n,1},\dots,A_{n,K_n}\}) = p(|A_{n,1}|,\dots,|A_{n,K_n}|)
\]
where the sets $A_{n,j}$ are considered in order of appearance and $p$ is a symmetric function of its arguments called \emph{exchangeable partition probability function} (EPPF). Therefore, by definition, the ordering in which we observe the data is not taken into account and the only information that affects the distribution of the random partition is the size of the clusters. For an infinite sequence of random variables $(\theta_{(1)},\theta_{(2)},\ldots)$ taking values in $\mathbb R_+$, let $\Pi(\theta_{(1)},\theta_{(2)},\ldots)$ be the random partition of $\mathbb N$ defined by the equivalence relation ``$i$ and $j$ are in the same cluster" if and only if $\theta_{(i)} = \theta_{(j)}$~\citep{Pitman1995}.  By Kingman's representation theorem~\citep{Kingman1975}, every exchangeable random partition has the same distribution as $\Pi(\theta_{(1)},\theta_{(2)},\ldots)$, where the random variables $\theta_{(1)},\theta_{(2)},\ldots$ are conditionally independent and identically distributed from some random probability distribution $\mathbb P$.

A popular model for this random probability distribution is a normalized completely random measure~\citep{Regazzini2003,James2005,James2009}, defined as $\mathbb P=W/W(\mathbb R_+)$, where $W\sim \CRM(\rho,\alpha)$ and $\alpha(\mathbb R_+)<\infty$. This condition, together with the condition \eqref{eq:conditionmeasure}, ensures that $0<W(\mathbb R_+)<\infty$  almost surely, and the model is thus properly defined.

\subsection{Continuous-time embedding of exchangeable random partitions via Poissonization}

Let $(\theta_{(1)},\theta_{(2)},\ldots)$ be an infinite sequence of random variables taking values in $\mathbb R_+$ and $\Pi(\theta_{(1)},\theta_{(2)},\ldots)$ be the random partition of $\mathbb N$ defined by the equivalence relation ``$i$ and $j$ are in the same cluster" if and only if $\theta_{(i)} = \theta_{(j)}$. Let $0< \tau_{(1)} < \tau_{(2)}<\ldots$ be an infinite sequence of arrival times. Define the continuous-time partition-valued process $(\Pi(t))_{t\geq 0}$ as
\begin{align*}
\Pi(t):=\Pi_{N(t)}=\Pi(\theta_{(1)},\ldots,\theta_{(N(t))}),~~t\geq 0
\end{align*} where $N(t)=\sum_{i} \1{\tau_{(i)}\leq t}$ with  $\1{\tau\leq t}=1$ if $\tau\leq t$ and 0 otherwise. $(\Pi(t))_{t\geq 0}$ defines a continuous-time embedding of the partition $\Pi$. Note that we have
\begin{align*}
\Pi_n=\Pi(\tau_{(n)})~~n\geq 1.
\end{align*}

A remarkable feature of exchangeable random partitions is that they admit a continuous-time embedding via a Poisson process. Poissonization is a classical technique used in combinatorial problems in order to derive analytical properties of exchangeable partitions and urn schemes~\citep{Karlin1967,Gnedin2007,Broderick2012}. Let $\mathbb P$ be a random probability measure on $\mathbb R_+$ associated to the exchangeable random partition $\Pi$. Consider the Poisson process $Q=\{(\tau_i,\theta_i)\}_{i\ge1}$ on $\mathbb{R}_+ \times \mathbb{R}_+$ with mean measure $d\tau\mathbb P(d\theta)$. Let $(\tau_{(i)},\theta_{(i)})_{i\ge1}$ be the sequence of points ordered by their arrival times, that is $\tau_{(1)}<\tau_{(2)}<\ldots$, and let $(\Pi(t))_{t\geq 0}$ be the associated continous-time partition-valued process. Then by construction, $(\Pi(t))_{t\geq 0}$ is a continuous-time embedding of the exchangeable random partition $\Pi$ of $\mathbb N$.

\begin{figure}[t]
  \begin{center}
    \subfigure[Poisson embedding of an exchangeable random partition]{\includegraphics[width=.49\textwidth]{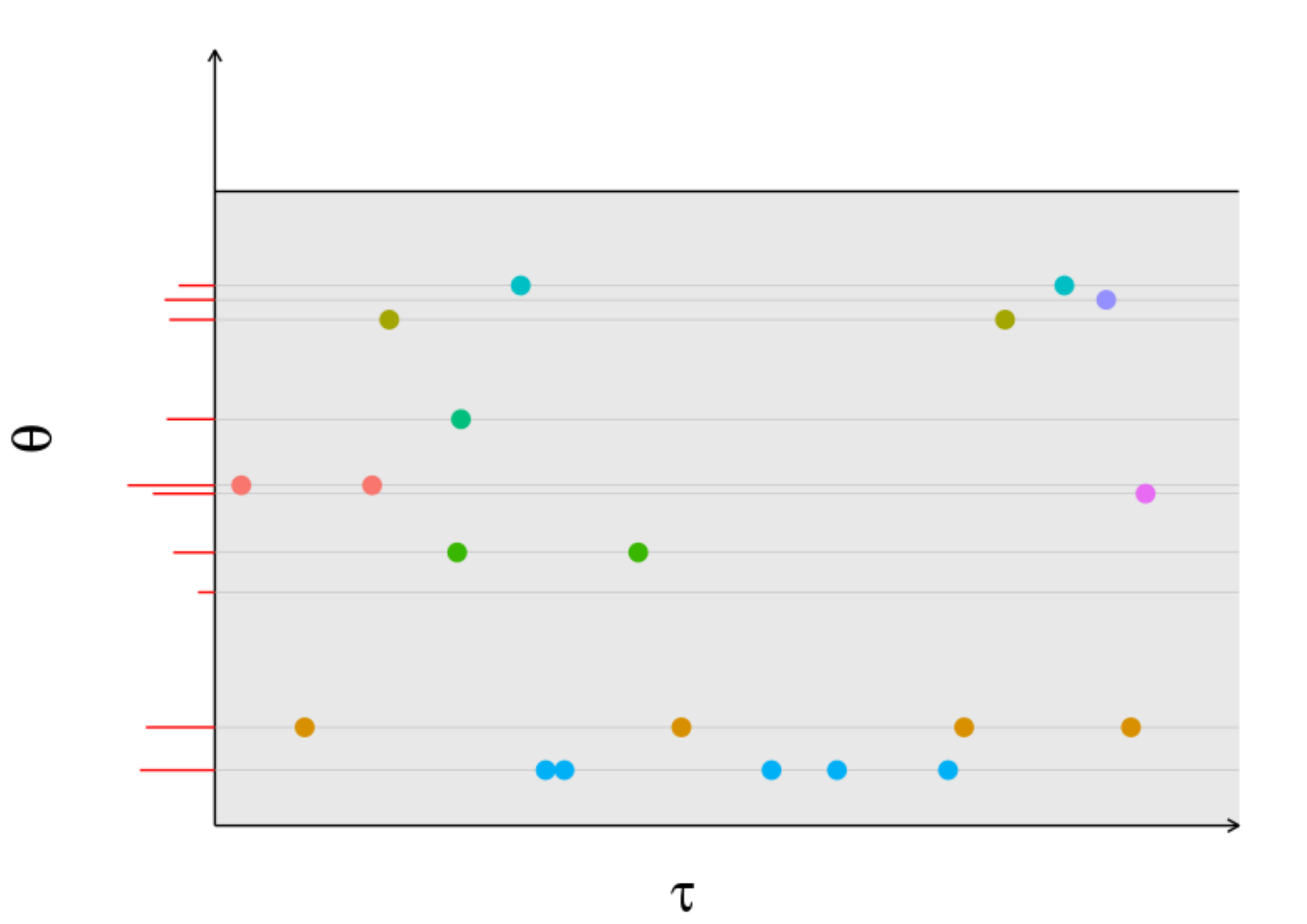}}
    \subfigure[Poisson embedding of the non-exchangeable random partition]{\includegraphics[width=.49\textwidth]{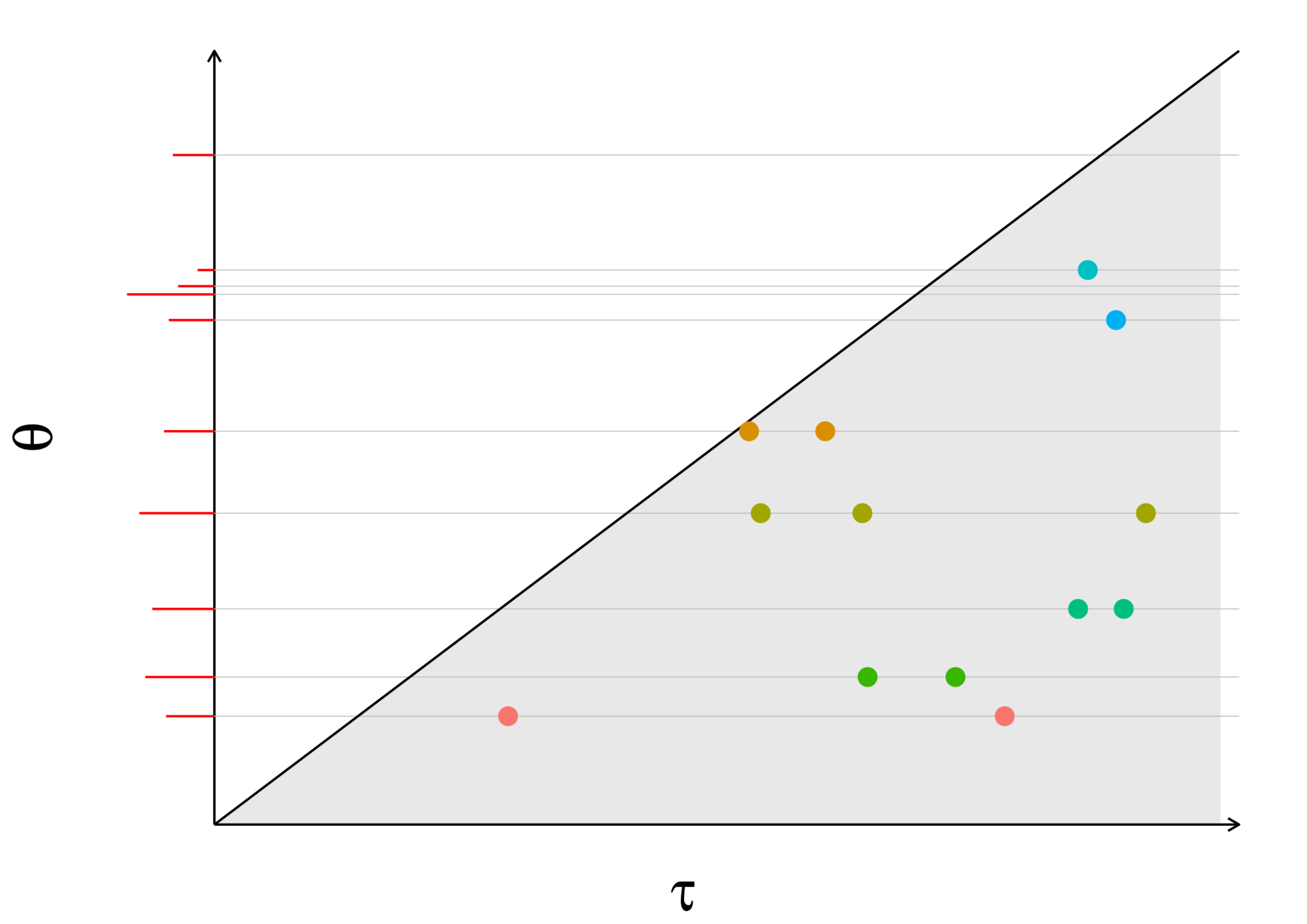}}
    \caption{(a) The Chinese restaurant process obtained via a Poisson embedding. Points $(\tau_i,\theta_i)$ are drawn from a Poisson point process on $\mathbb R_+\times[0,1]$ with mean measure $\mu(d\tau,d\theta) = d\tau\,W(d\theta)\1{\theta\leq1}$ where $W$ is a gamma random measure with base measure $\alpha(d\theta)=\alpha_0 d\theta$ and $\zeta=1$. The red sticks on the $\theta$-axis represent the jumps of the Gamma random measure $W$. Points on the same horizontal line are in the same cluster. The random partition $\Pi(\theta_{(1)},\theta_{(2)},\ldots)$ of $\mathbb N$ induced by the sequence of points $\theta_{(1)},\theta_{(2)},\ldots$ ordered by their arrival times $\tau_{(1)}<\tau_{(2)}<\ldots$ is the CRP. (b) Non-exchangeable random partition model via a Poisson embedding. Points $(\tau_i,\theta_i)$ are drawn from a Poisson point process on $\mathbb R_+\times[0,1]$ with mean measure $\mu(d\tau,d\theta) = d\tau\,W(d\theta)\1{\theta\leq\tau}$ where $W$ is a CRM. The red sticks on the $\theta$-axis represent the jumps of the CRM $W$. Points on the same horizontal line are in the same cluster.}
    \label{fig:poissonembedding}
  \end{center}
\end{figure}

We now focus on the important case where the partition is obtained from a normalized CRM~\citep{Regazzini2003,Nieto-Barajas2004,Lijoi2005,Lijoi2005a,Lijoi2007, James2009}, which includes as a special case the one-parameter CRP. Let $Q=\{(\tau_i,\theta_i)\}_{i\ge1}$ be a Poisson (Cox) process on $\mathbb{R}_+ \times \mathbb{R}_+$ with random mean measure
\[
\mu(d\tau,d\theta) = d\tau\,W(d\theta)\1{\theta\leq1}
\]
where $W\sim\CRM(\alpha,\rho)$  with base measure $\alpha(d\theta) = \alpha_0\,d\theta$ with $\alpha_0>0$ and L\'evy measure $\rho$ satisfying $\int_0^\infty \rho(d\omega)=\infty$. The point process $Q$ has support on $\mathbb{R}_+\times[0,1]$ and we can write it as follows
\begin{align*}
  W &\sim \CRM\,(\alpha,\rho)\\
  Q\,|\,W &\sim \Poi(d\tau\,W(d\theta)\1{\theta\leq1})
\end{align*}
where $\Poi(\mu)$ denotes a Poisson point process with mean $\mu$.  This is illustrated in Figure~\ref{fig:poissonembedding}(a). The distribution of the first $n$ points  given $W$ is
\begin{equation}
\Pr(d\theta_{(1:n)},d\tau_{(1:n)}\mid W)=\left[  \prod_{i=1}^{n}W\left\{d\theta_{(i)}\right \}\right]  e^{-\tau_{(n)}\overline W(1)%
}\1{\tau_{(1)}<\tau_{(2)}<\ldots<\tau_{(n)}}\,d\tau_{(1:n)}\label{eq:likelihood}
\end{equation}
where $\overline W(t)=\int_0^t W(d\theta)=\sum_{i\geq 1} \omega_i \1{\vartheta_i\leq t}$.
Let us denote by $m_{n,j}$ the number of points in the $j$-th cluster after having observed $n$ points, and by $(\theta_i^*)_{i=1,\dots,K_n}$ the unique values in $(\theta_{(1)},\ldots,\theta_{(n)})$, ordered by arrival times. Using the results in
\cite[Proposition 3.1 page 18]{James2002}, we can obtain the expectation of \eqref{eq:likelihood} with respect to the CRM $W$
\[
\Pr(d\theta_{(1:n)},d\tau_{(1:n)})=e^{-\alpha_0\psi(\tau_{(n)})}\alpha_0^{K_n}\left[  \prod_{j=1}^{K_{n}}\kappa(m_{n,j},\tau_{(n)})\,d\theta_j^*\right]
\1{\tau_{(1)}<\ldots<\tau_{(n)}}\,d\tau_{(1:n)}.
\]
Integrating over the arrival times $\tau_{(i)}$ and the cluster locations $\theta_{j}^{\ast}$ gives
\[
\Pr(\Pi_{n})=\int_{0}^{\infty}e^{-\alpha_0\psi(u)}\alpha_0^{K_n}\left[  \prod
  _{j=1}^{K_{n}}\kappa(m_{n,j},u)\right]  \frac{u^{n-1}}{\Gamma(n)}du,
\]
and one recovers the EPPF of the exchangeable random partition associated to a normalized completely random measure \citep[Corollary 6]{Pitman2003}, \cite[Proposition 3]{James2009}.
In the gamma process case, $\kappa(m,u)=\Gamma(m)/(1+u)^m$ and $\psi(u)=\log(1+u)$, yielding
\[
\Pr(\Pi_{n})=\frac{\alpha_0^{K_{n}}}{\Gamma(n)}\left[  \prod
_{j=1}^{K_{n}}\Gamma(m_{n,j})\right]\int_0^\infty \frac{u^{n-1}}{(1+u)^{n+\alpha_0}}du=\\
\frac{\alpha_0^{K_{n}}\Gamma(\alpha_0)}{\Gamma(\alpha_0+n)}\left[  \prod
_{j=1}^{K_{n}}\Gamma(m_{n,j})\right]
\]
which is the EPPF of the Chinese Restaurant process.

\section{Non-exchangeable random partitions}
\label{sec:model}

In this section, we build on the Poissonization idea in order to derive a class of non-exchangeable random partitions. This class is shown to have the microclustering property in the next section. The Cox Process $Q=\{(\tau_i,\theta_i)\}_{i\geq 1}$ that defines our non-exchangeable random partition model has the following random mean measure
\begin{equation}\label{eq:mm}
\mu(d\tau, d\theta) = \1{\theta \leq \tau} W(d\theta) d\tau
\end{equation}
therefore the points will lie under the bisector as shown in Figure~\ref{fig:poissonembedding}(b). The overall model is therefore defined as
\begin{align*}
  W &\sim \CRM\,(\alpha,\rho)\\
  Q\,|\,W &\sim \Poi(\1{\theta\leq\tau}\,d\tau\,W(d\theta))
\end{align*}

The random partition $\Pi=(\Pi_n)_{n\geq 1}$ of $\mathbb N$ is obtained by considering the points $((\tau_{(i)},\theta_{(i)}))_{i\geq 1}$ of the point process $Q$ ordered by their arrival time, and let $\Pi_n=\Pi(\theta_{(1)},\ldots,\theta_{(n)})$ be the partition induced by the first $n$ points for any $n\geq 1$. The random partition model is completely specified by the base measure $\alpha$ and the L\'evy measure $\rho$.

The crucial difference with the previous construction is the support of the point process. In the continuous time version of the CRP, every atom of $W$ in $[0,1]$ was allowed to be chosen at any time, hence the set of potential cluster labels was constant over time. Now, for every fixed $t>0$, all the clusters whose $\theta$ are greater than $t$ cannot be chosen before that time, therefore the set of potential cluster labels increases with $t$, if for instance the base measure $\alpha$ has unbounded support on $\mathbb{R}_+$. This property intuitively leads to both the non-exchangeability of the random partition induced on $\mathbb{N}$, but also to the microclustering property.

Samples from the process are represented in Figure~\ref{fig:pointprocess} when $W\sim \GGP(\alpha,\sigma,1)$ with base measure $\alpha(d\theta)=\xi\theta^{\xi-1}$, for different values of $\xi$ and $\sigma$.

\begin{figure}[h!]
\centering
\includegraphics[width=.43\textwidth]{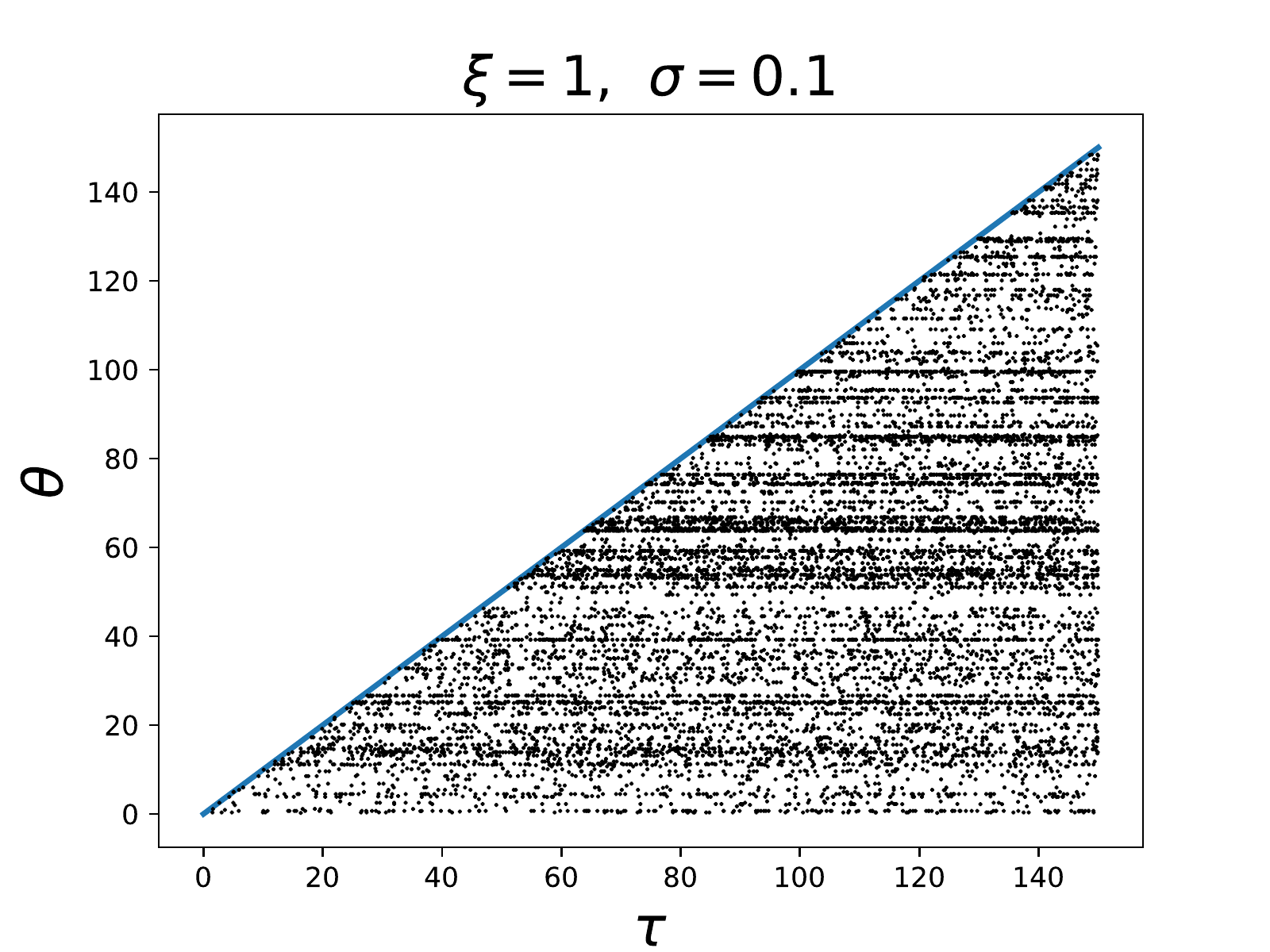} \quad \includegraphics[width=.43\textwidth]{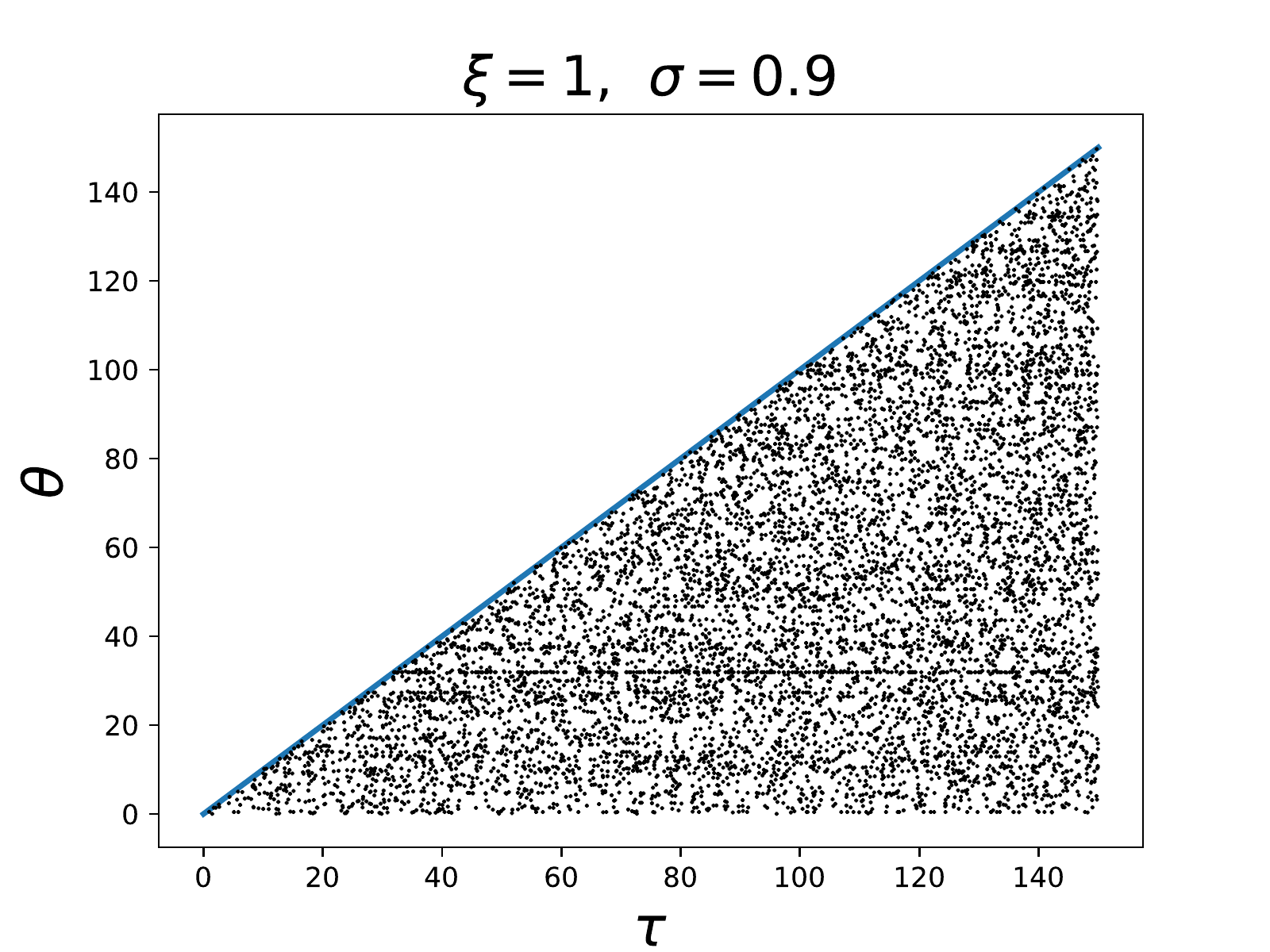} \\
\includegraphics[width=.43\textwidth]{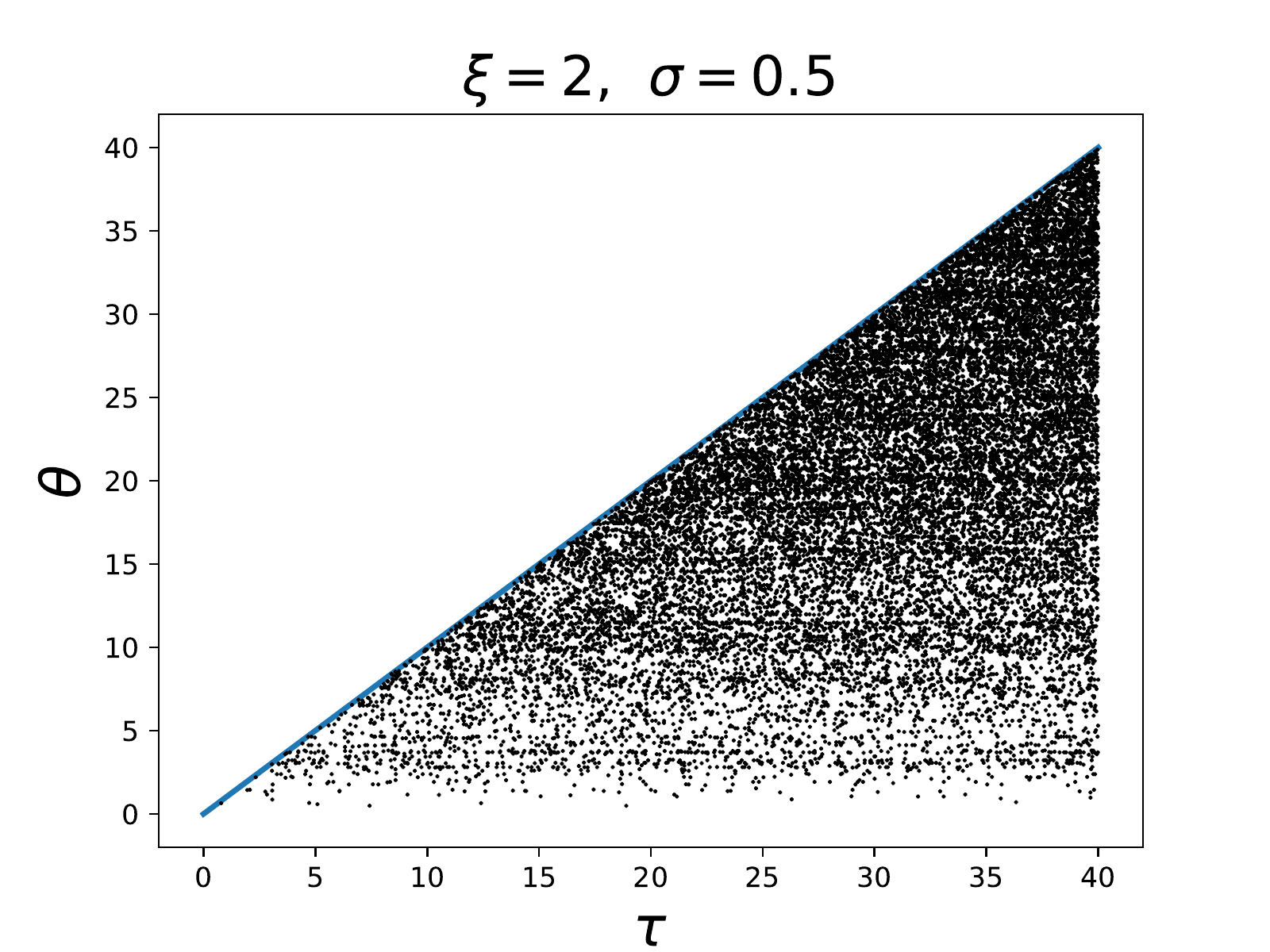} \quad \includegraphics[width=.43\textwidth]{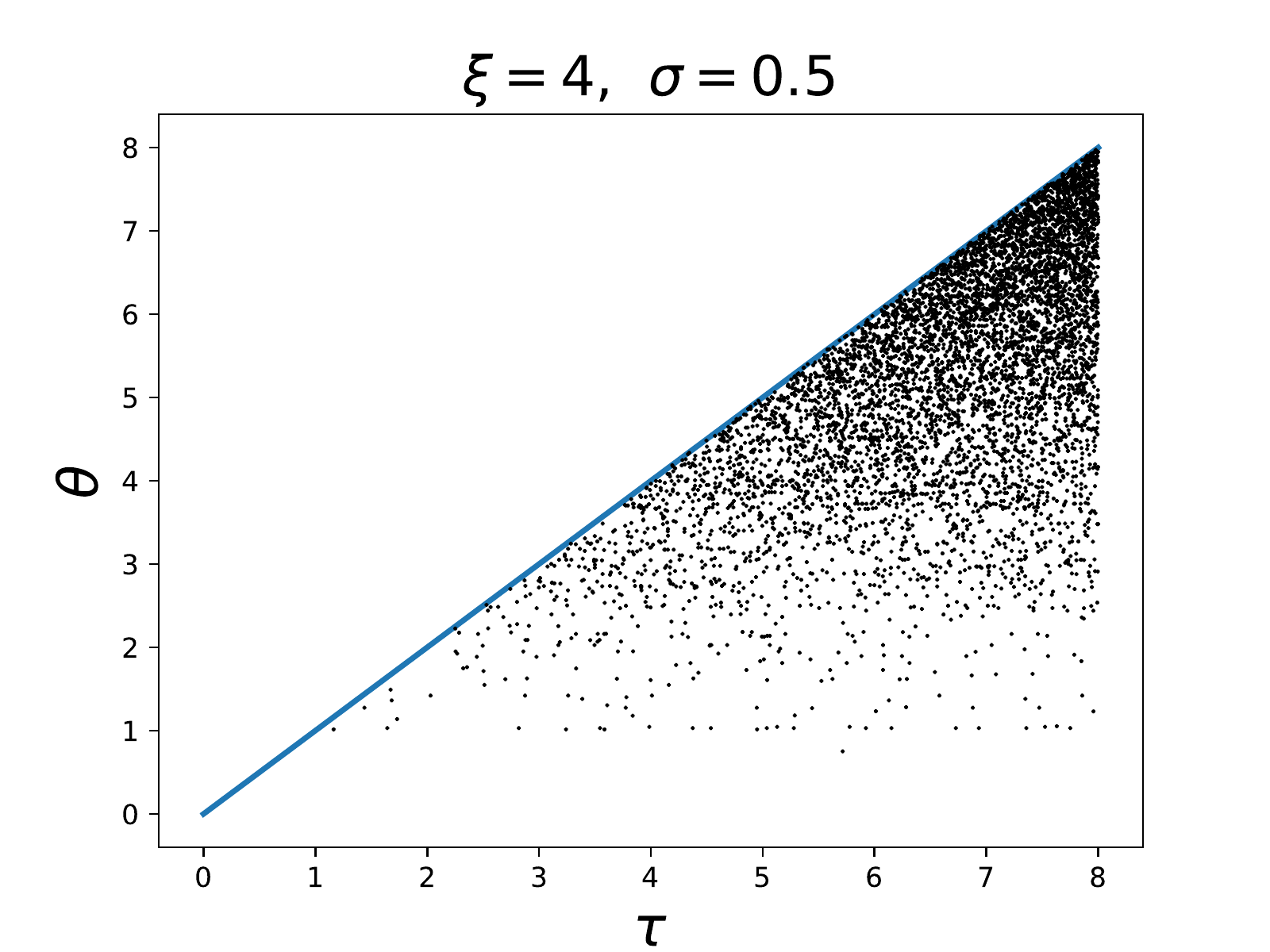}
\caption{Samples from the Cox Process $Q$ where $W\sim GGP(\alpha,\sigma,1)$ with base measure $\alpha(d\theta) = \xi\,\theta^{\xi-1}\,d\theta$, for different values of $\sigma$ and $\xi$.}
\label{fig:pointprocess}
\end{figure}

\begin{proposition}\label{likelihood}
  Let $W$ be a homogeneous CRM$(\alpha,\rho)$ and a point process $Q=\{(\tau_i,\theta_i)\}_{i\geq 1}$ on $\mathbb R_+^2$ with mean measure $\mu(d\tau\,d\theta)=\mathbbm{1}_{\theta\le\tau}d\tau\,W(d\theta)$. Let $((\tau_{(i)},\theta_{(i)}))_{i\geq 1}$ be the sequence of points ordered in time, that is such that $\tau_{(1)}<\tau_{(2)}<\ldots$. For any $n\geq 1$,

\begin{align}
\Pr(d\theta_{(1:n)},d\tau_{(1:n)})&=\left[\prod_{j=1}^{K_{n}}%
\kappa(m_{n,j},\tau_{(n)}-\theta_{j}^{\ast})\alpha(d\theta_{j}^{\ast})\right]
e^{-\int_0^{\tau_{(n)}} \psi(\tau_{(n)}-\theta)\alpha(d\theta)}\nonumber\\
&\times\quad\left[  \prod_{i=1}^{n}\1{\theta_{(i)}\leq\tau_{(i)}}\right]  \1{\tau_{(1)}<\tau_{(2)}<\ldots<\tau_{(n)}}%
\,d\tau_{(1:n)}\label{eq:joint}
\end{align}
where $\theta_{j}^{\ast}$, $j=1,\ldots,K_n$, are the unique values of $(\theta_{(1)},\ldots,\theta_{(n)})$, and $m_{n,j}$ their multiplicities.
\end{proposition}
\begin{proof}
The derivation is similar to the derivation for the exchangeable case described in the previous section. Given $W$, the set of points $(\tau_i)_{i\geq 1}$ is an inhomogeneous Poisson point process on $\mathbb R_+$ with rate $\overline W(t)$, hence
 \begin{align*}
    \Pr(d\tau_{(1:n)}\,|\,W) = e^{-\int_{0}^{\tau_{(n)}}
\overline W(t)dt}\left[\prod_{i=1}^n \overline W(\tau_{(i)})\right] \1{\tau_{(1)}<\ldots<\tau_{(n)}} d\tau_{(1:n)}.
  \end{align*}
  Given the $n$ time variables, the $\theta_{(i)}$'s are distributed as follows
  \begin{align*}
    \Pr(d\theta_{(1:n)}\,|\,\tau_{(1:n)}, W) =\prod_{i=1}^n\frac{W(d\theta_{(i)})}{\overline W(\tau_{(i)})} \1{\theta_{(i)}<\tau_{(i)}}.
  \end{align*}
It follows that
  \begin{align}
    \Pr(d\theta_{(1:n)},d\tau_{(1:n)}\mid W)&=\left[  \prod_{i=1}^n W(d\theta_{(i)})\1{\theta_{(i)}\leq \tau_{(i)}}\right]  e^{-\int_{0}^{\tau_{(n)}}
\overline W(t)dt} \1{\tau_{(1)}<\ldots<\tau_{(n)}}\,d\tau_{(1:n)}\label{eq:likelihoodnonex}
  \end{align}

where  $\int_{0}^{\tau_{(n)}}\overline W(t)dt=\sum_{j}\omega_{i}(\tau_{(n)}-\vartheta_{j})_{+}=W(g_{\tau_{(n)}})$ with $g_{t}(x)=(t-x)_+=\max(0,t-x)$.
 Using \cite[Proposition 3.1]{James2002}, we can integrate over $W$ to obtain the final result.
\end{proof}

Integrating Equation~\eqref{eq:joint} over the cluster allocations $(\theta_j^*)_{j=1,\dots,K_n}$ and
the arrival times $\tau_{(1:n)}$, we would obtain the distribution of the random partition $\Pi_n$. To the best of our knowledge, there is however no analytical expression for this distribution. We can nonetheless simulate random partitions by using the cluster allocations and arrival times as latent variables. In particular, for the generalized gamma process, we have the following result.

\begin{proposition}\label{predictive}
  Let $W\sim\text{GGP}\,(\alpha,\sigma_0,\zeta_0)$, and
  $Q=\{(\tau_i,\theta_i)\}_{i\ge1}$ be the points of a Cox process with
  mean measure $\mu(d\tau,d\theta) = \1{\theta \leq \tau}
  W(d\theta) d\tau$. Then the predictive distribution of $\tau_{(n)}$ has density

\begin{align*}
p(\tau_{(n)}& \,|\, (\theta_{(i)}, \tau_{(i)})_{i=1,\dots, n-1}) \propto
\left[ \prod_{j=1}^{K_{n-1}} \frac{1}{(\tau_{(n)}-\theta_j^* +\zeta_0)^{m_{n-1,j}-\sigma_0}}\right]e^{-\int_0^{\tau_{(n)}} \psi(\tau_{(n)}-\theta)\alpha(d\theta)}\\ &\times\Bigg(\sum_{j=1}^{K_{n-1}}\frac{m_{n-1,j}-\sigma_0}{\tau_{(n)}-\theta_j^*+\zeta_0}+\int_0^{\tau_{(n)}}\frac{\alpha(\theta)}{(\tau_{(n)} - \theta + \zeta_0)^{1-\sigma_0}}\, d\theta \Bigg)\, \1{\tau_{(n)}>\tau_{(n-1)}}
\end{align*}
where $\psi(t) = \log(1+t/\zeta_0)$ for $\sigma_0=0$, while $\psi(t) = \left((t+\zeta_0)^{\sigma_0}-\zeta_0^{\sigma_0}\right)/\sigma_0$ for $\sigma_0\in(0,1)$. The conditional distribution for $\theta_{(n)}$
  is a convex combination of a discrete distribution and a diffuse one,
 \[
 \Pr(\theta_{(n)}\in d\theta \,|\,
 (\theta_{(i)}, \tau_{(i)})_{i=1,\dots, n-1},\tau_{(n)}) \propto
 H_{\tau_{(n)}}(d\theta) +
 \sum_{i=1}^{K_{n-1}}
 \frac{m_{n-1,i}-\sigma_0}{\tau_{(n)}-\theta_i^* + \zeta_0}
 \delta_{\theta_i^*}(d\theta)
 \]
 where $H_t$ is a diffuse distribution defined as
 $H_t(A) = \int_A
 \frac{\1{\theta\leq t}}
      {(t-\theta +\zeta_0)^{1-\sigma_0}}\,\alpha(d\theta)$
      for every Borel set $A\subset \mathbb{R}_+$.
\end{proposition}

For example, if $W$ is a gamma process ($\sigma_0=0$) and $\alpha(d\theta)=d\theta$, we obtain
\begin{align*}
p(\tau_{(n)} \,|\, (\theta_{(i)}, \tau_{(i)})_{i=1,\dots, n-1}) &\propto
\left[ \prod_{j=1}^{K_{n-1}} \frac{1}{(\tau_{(n)}-\theta_j^* +\zeta_0)^{m_{n-1,j}}}\right]e^{-\tau_{(n)}}\left(1+\frac{\tau_{(n)}}{\xi_0}\right)^{-\tau_{(n)}-\xi_0} \\ &\times\Bigg(\sum_{j=1}^{K_{n-1}}\frac{m_{n-1,j}}{\tau_{(n)}-\theta_j^*+\zeta_0}+\log(1+\tau_{(n)}/\zeta_0) \Bigg)\, \1{\tau_{(n)}>\tau_{(n-1)}},
\end{align*}
and
 \begin{align*}
 \left \{
 \begin{tabular}{ll}
   $\Pr(\theta_{(n)}=\theta_j^*~~\mid (\theta_{(i)}, \tau_{(i)})_{i=1,\dots, n-1},\tau_{(n)})=C_n\frac{m_{n-1,j}}{\tau_{(n)}-\theta_j^* + \zeta_0}$ & for $j=1,\ldots,K_{n-1}$ \\
   $\Pr(\theta_{(n)}\text{ is new}\mid (\theta_{(i)}, \tau_{(i)})_{i=1,\dots, n-1},\tau_{(n)})=C_n\log(1+\tau_{(n)}/\zeta_0)$
 \end{tabular}
 \right .
\end{align*}
 where $C_n$ is the appropriate normalizing constant.

\section{Properties and inference}
\label{sec:properties}

\subsection{Asymptotic properties}

 In this section, denote $X_t\sim Y_t$, $X_t=o(Y_t)$ and $X_t=O(Y_t)$ respectively for $X_t/Y_t\rightarrow 1$, $X_t/Y_t\rightarrow 0$ and $\lim\sup_t X_t/Y_t<\infty$. The notation $X_t\asymp Y_t$ means both $X_t=O(Y_t)$ and $Y_t=O(X_t)$ hold. When $X_t$ and/or $Y_t$ are random variables the asymptotic relation is meant to hold almost surely.\smallskip

The properties we are most interested in are the asymptotic behaviour of the cluster sizes $m_{n,j}$, of the number $K_n$ of clusters and the number $K_{n,r}$ of clusters of size $r$ in the random partition. We show in this section that our non-exchangeable model allows for a sublinear growth of the clusters' sizes while retaining desirable properties for the other quantities. Let us list the assumptions on the CRM $W$ to derive the asymptotic results.
\begin{itemize}
\item[(A1)] $W$ has finite first two moments, that is
\[
\kappa(1,0)=\int_0^\infty\omega\rho(d\omega)<\infty\text{ and } \kappa(2,0)=\int_0^\infty\omega^2\rho(d\omega)<\infty.
\]
\item[(A2)] The L\'evy tail intensity $\bar{\rho}(x)=\int_x^\infty \rho(d\omega)$ is a \emph{regularly varying} function at 0, that is
\[
\overline{\rho}(x)\sim \ell\left(1/x\right)x^{-\sigma}
\]
as $x\rightarrow 0^+$, where $\ell$ is a slowly varying function at infinity and $\sigma\in[0,1]$.
\item[(A3)] The improper cumulative distribution $\overline\alpha(t)=\int_0^t \alpha(dx)$ of the base measure $\alpha$ is a regularly varying function at infinity, that is
$$
\overline\alpha(t)\sim L(t)\,t^\xi
$$
as $t\rightarrow\infty$, where $\xi>0$ and $L$ is a slowly varying function. Assume additionally that the base measure $\alpha$ is dominated by the Lebesgue measure, and admits a continuous and monotone density denoted $\alpha(\theta)$.
\end{itemize}
The moment assumption (A1) excludes the stable process that has infinite first moment. (A2) controls, through the parameter $\sigma$, the power-law behaviour of the proportion of clusters of a given size, while condition (A3) is used to prove the microclustering property and control the sublinear rate of the clusters' size. It is worth noting that the last condition is very mild and allows to pick the density of the base measure from a very large class of functions. Assumptions (A1-A2) are satistied for the GGP with parameters $\sigma_0\in(-\infty,1)$ and $\zeta_0> 0$. In this case, we have $\sigma=\max(\sigma_0,0)$ and $\ell(t)\propto \log t$ for $\sigma=0$ and $\ell(t)$ is constant otherwise.  \bigskip

Recall that $$N(t)=\sum_{i\geq 1} \1{\tau_i\leq t}$$ denotes the number of points of $Q$ such that $\tau_i\leq t$. For each atom $\vartheta_j$, $j\geq 1$, of the CRM $W$, let $$X_j(t)=\sum_{i\geq 1} \1{\tau_i\leq t}\1{\theta_i=\vartheta_j}.$$
For $j\geq 1$, let $$M_j(t)=\sum_{i\geq 1} \1{\tau_i\leq t}\1{\theta_i=\theta^\ast_j}$$ the size of cluster $j$, ordered by appearance, at time $t$. Note that $N(\tau_{(n)})=n$ and $M_j(\tau_{(n)})=m_{n,j}$.

\begin{proposition}\label{thm:clustsize}
  Let $W=\sum_{j\geq 1}\omega_i \delta_{\vartheta_j}$ be a CRM with mean measure $\alpha(d\theta)\rho(d\omega)$ satisfying Assumptions (A1-A3). Let $\{(\tau_i,\theta_i)\}_{i\geq 1}$ be a Poisson point process with mean measure $\mu(d\tau\,d\theta)=\1{\theta\leq \tau}d\tau\,W(d\theta)$. We have, almost surely as $t$ tends to infinity,
  \begin{align*}
  N(t)&\sim \frac{\kappa(1,0)}{\xi+1}t^{\xi +1}L(t)
\end{align*}
and, for $j\geq 1$
  \begin{align*}
 X_j(t)&\sim W(\{\vartheta_j\}) t\\
 M_j(t)&\sim W(\{\theta_j^\ast\})t.
\end{align*}
\end{proposition}

Proposition \ref{thm:clustsize} implies the microclustering property for the random partition $\Pi_n$: almost surely, $M_j(t)/N(t)\rightarrow 0$ as $t\rightarrow \infty$, hence $m_{n,j}/n\rightarrow 0$ as $n\rightarrow\infty$. In the following corollary, which follows from properties of inverse of regularly varying functions~\citep[Proposition 1.5.15]{Bingham1987} or \citep[Lemma 22]{Gnedin2007}, we obtain exact rates of growth for the cluster sizes.

\begin{corollary}[Microclustering property]
We have
\begin{equation}
t\sim \left (\frac{\xi +1}{\kappa(1,0)}\right )^{1/(\xi+1)} L_{\xi+1}^\ast(N(t)) \,N(t)^{1/(\xi+1)}\label{eq:inversiont}
\end{equation}
almost surely as $t\rightarrow \infty$, where $L_{\xi+1}^*$ is a slowly varying function defined in equation \eqref{eq:Lstar} in the Appendix. It follows that the cluster sizes $m_{n,j}=M_j(\tau_{(n)})$ verify, for any $j\geq 1$, $$m_{n,j}\sim W(\{\theta_j^\ast\}) \left (\frac{\xi +1}{\kappa(1,0)}\right )^{1/(\xi+1)} L_{\xi+1}^\ast(n)\,n^{1/(\xi+1)} $$
almost surely as $n$ tends to infinity. For $j=1$, the distribution of $\omega_1^*=W(\{\theta^*_1\})$ is given by
\begin{align*}
\Pr(d\omega^\ast_1)=\omega_1^\ast \rho(d\omega_1^\ast)\int_0^\infty\int_0^\tau e^{-\omega_1^\ast (\tau-\theta)}
e^{-\int_0^t \psi(\tau-u )\alpha(du) }
\alpha(d\theta)d\tau.
\end{align*}
\end{corollary}
Note that the growth rate of the cluster sizes only depends on the parameters $\xi$ and $L$ of the base measure $\alpha$, and not on the properties of the L\'evy measure $\rho$.  For example, taking $\overline\alpha(t)=\gamma t^\xi$, with $\xi,\gamma>0$, we have $L(t)=\gamma$ and $m_{n,j}\asymp n^{1/(\xi +1)}$ and the cluster sizes grow at a rate of $n^a$ where $0<a<1$.

We now provide results on the asymptotic rates of the number of clusters and number of clusters of a given size, showing that we can have the same range of behaviour as for exchangeable random partitions.
Let $$K(t)=\sum_{j\geq 1}\1{X_j(t)>0}$$ be the number of different clusters in $\Pi(t)$ at time $t$ and $$K_r(t)=\sum_{j\geq 1}\1{X_j(t)=r}$$ the number of clusters of size $r$ at time $t$.

\begin{proposition}\label{thm:nbclust}
  Let $W=\sum_{j\geq 1}\omega_i \delta_{\vartheta_j}$ be a CRM with mean measure $\alpha(d\theta)\rho(d\omega)$ satisfying Assumptions (A1-A3). Define
   $$\left \{
   \begin{array}{ll}
     \ell_\sigma(t)=\Gamma(1-\sigma)\ell(t) & if  \sigma\in[0,1)\\
     \ell_1(t)=\int_t^\infty y^{-1}\ell(y)dy & if \sigma=1.
   \end{array}  \right .
   $$
   Let $\{(\tau_i,\theta_i)\}_{i\geq 1}$ be a Poisson point process with mean measure $\mu(d\tau\,d\theta)=\1{\theta\leq \tau}d\tau\,W(d\theta)$. We have, almost surely at $t$ tends to infinity,
  \begin{align*}
 K(t) \sim \frac{\Gamma(\sigma+1)\Gamma(\xi+1)}{\Gamma(\sigma+\xi+1)}\,L(t)\ell_\sigma(t)\, t^{\sigma+\xi}.
 \end{align*}
 For $r\geq 1$, if $\sigma=0$ then $K_r(t)=o(K(t))$, if $\sigma\in(0,1)$,
 \begin{align*}
  K_r(t)\sim\frac{\sigma\Gamma(r-\sigma)}{r!\Gamma(1-\sigma)}K(t)\, .
  \end{align*}
    If $\sigma=1$, $K_1(t)\sim K(t)$ and $K_r(t)=o(K(t))$ for all $r\geq 2$.

\end{proposition}
By noting that $K_n=K(\tau_{(n)})$ and $K_{n,r}=K_r(\tau_{(n)})$, we can combine the results of Proposition~\ref{thm:nbclust} and Equation~\eqref{eq:inversiont} to obtain asymptotic expressions for the number $K_n$ of clusters and the number $K_{n,j}$ of clusters of size $j$ in $\Pi_n$.
  \begin{corollary}\label{corollarypowerlaw}
  We have, almost surely as $n$ tends to infinity,
  \begin{align*}
 K_n \sim \,\widetilde{\ell}(n)\, n^{(\sigma+\xi)/(\xi+1)}
 \end{align*}
 where $\widetilde \ell$ is a slowly varying function defined in equation \eqref{eq:elltilde} in the Appendix.

 For $r\geq 1$, if $\sigma=0$ then $K_{n,r}=o(K_n)$; if $\sigma\in(0,1)$,
 \begin{align*}
  \frac{K_{n,r}}{K_n}\rightarrow\frac{\sigma\Gamma(r-\sigma)}{r!\Gamma(1-\sigma)}\, .
  \end{align*}
  This corresponds to a power-law behaviour for the proportion of clusters of size $r$, as $$\frac{\sigma\Gamma(r-\sigma)}{r!\Gamma(1-\sigma)}\asymp \frac{1}{j^{1+\sigma}}$$ for large $j$. If $\sigma=1$, $K_{n,1}\sim K_n$ and $K_{n,r}=o(K_n)$ for all $r\geq 2$. In this case, the proportion of clusters of size 1 tends to one almost surely.
 \end{corollary}

\begin{example}
If $W\sim\text{GGP}(\alpha,\sigma,1)$ with $\sigma\in(0,1)$ and base measure $\alpha(d\theta) = \gamma \xi \theta^{\xi-1}d\theta$ with $\xi,\gamma>0$  we have $\ell(t) =\frac{1}{\sigma\,\Gamma(1-\sigma)}$ and $L(t) = \gamma $, therefore
\[
K_n \sim \frac{\Gamma(\sigma+1)\Gamma(\xi+1)}{\sigma\Gamma(\sigma+\xi+1)} (\xi+1)^{\frac{\sigma+\xi}{1+\xi}}\,
\gamma^{1-\frac{\sigma+\xi}{1+\xi}}\,
n^{\frac{\sigma+\xi}{1+\xi}}
\]
and for all $r\geq 1$
 \begin{align*}
  \frac{K_{n,r}}{K_n}\rightarrow\frac{\sigma\Gamma(r-\sigma)}{r!\Gamma(1-\sigma)}
  \end{align*}
almost surely as $n$ tends to infinity. This power-law behavior is illustrated on Figure~\ref{fig:powerlaw}.

\end{example}

It is worth noting that although the asymptotic behaviour of the number of clusters and the number of clusters of a given size depend also on the base measure $\alpha$, the power-law exponent in the proportion of clusters of a given size is solely tuned by the L\'{e}vy measure $\rho$ through the parameter $\sigma$.

\begin{figure}
\begin{center}
\includegraphics[scale=.2]{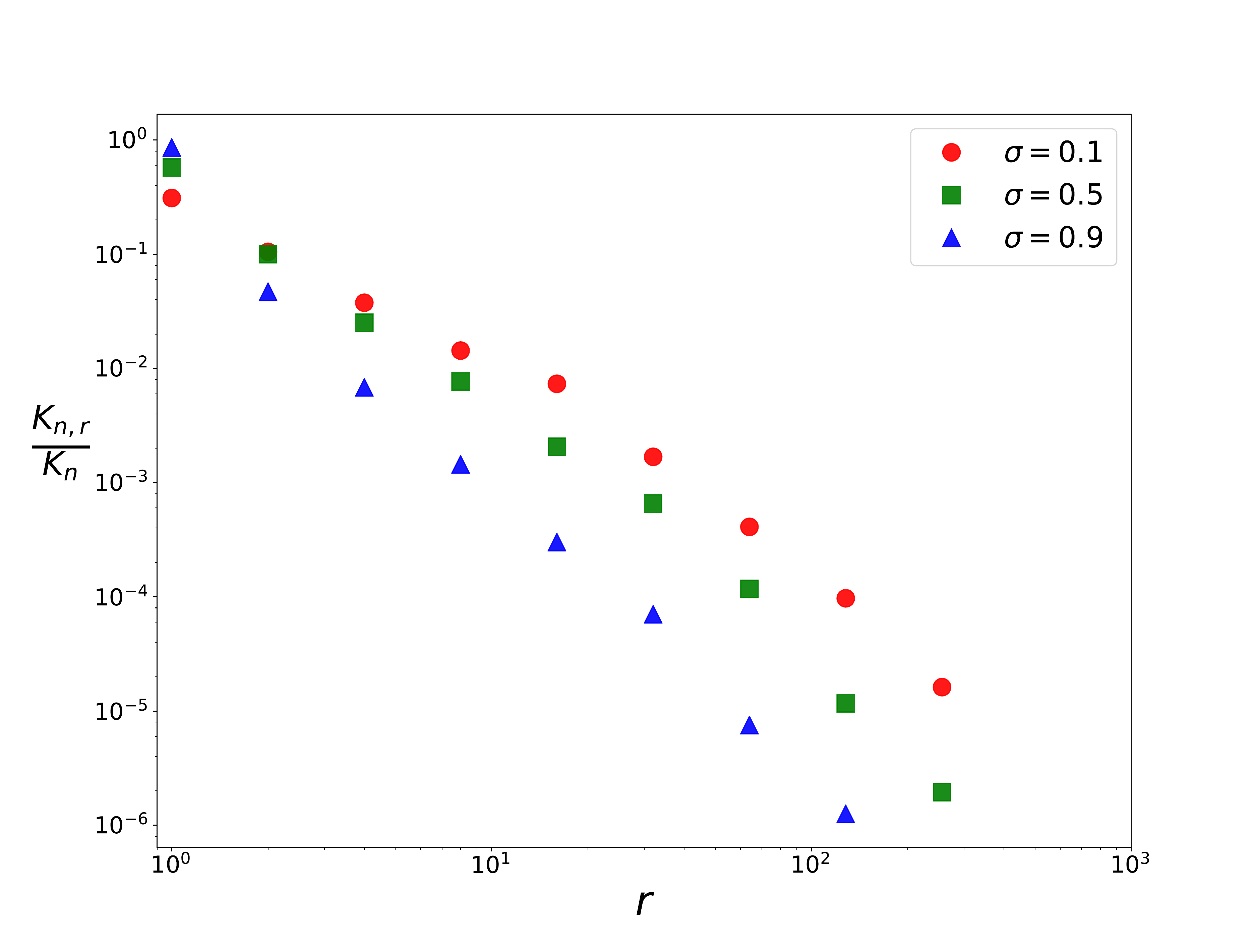}
\end{center}
\caption{Log-log plot of the proportions of clusters of given size for the GGP with $\alpha(d\theta)=d\theta$, $\zeta_0=1$, $\sigma=0.1,0.5,0.9$ and sample size 10000.}
\label{fig:powerlaw}
\end{figure}

\subsection{Inference} \label{inference}
\subsubsection{Posterior characterization}
Assuming we observe the first $n$ time-ordered points $(\tau_{(i)},\theta_{(i)})_{i=1,\ldots,n}$ from the Cox process $Q$, we want to characterize the conditional  distribution of the CRM $W$ given the time-ordered observations. The following posterior characterization follows from \citep[Proposition 3.1 page 18]{James2002}.

\begin{proposition} \label{posterior}
Given the first $n$ time-ordered observations $(\tau_{(i)},\theta_{(i)})_{i=1,\ldots,n}$ from the Cox process $Q$, with unique cluster labels $\theta^*_1,\dots,\theta^*_{K_n}$, the conditional distribution of the CRM $W$ is given by
\[
W' + \sum_{j=1}^{K_n}\omega^*_j\delta_{\theta^*_j}
\]
where the random positive weights $(\omega^*_1,\ldots,\omega^*_{K_n})$ are independent of the random measure $W'$. $W'$ is an inhomogeneous CRM with mean measure $\nu'(d\omega,d\theta) = e^{-\omega (\tau_{(n)}-\theta)_+}\rho(d\omega)\alpha(d\theta)$. The masses of the fixed atoms are conditionally independent with density
\[
p(\omega^*_j\mid \text{rest})\propto \rho(d\omega^*_j)\,\omega_j^{*m_{n,j}}e^{-\omega^*_j (\tau_{(n)}-\theta_j^*)} .
\]
In particular, when $W$ is a generalized gamma process the masses are conditionally gamma distributed
\[
\omega^*_j\mid \text{rest}\sim \Gam(m_{n,j}-\sigma_0, \zeta_0+\tau_{(n)}-\theta_j^*) .
\]
\end{proposition}

\subsubsection{Parameter estimation and prediction}
We consider the CRM with base measure $\alpha(d\theta) = \xi\,\theta^{\xi-1}d\theta$ and generalized gamma L\'evy measure with parameters $\sigma_0$ and $\zeta_0$. The set of parameters is therefore $\eta= (\xi,\sigma_0,\zeta_0)$. Having observed a partition $\Pi_n$, we aim at estimating the parameters $\eta$ and predict $\Pi_{n+m}$ for $m\geq 1$. The marginal likelihood $\Pr(\Pi_n|\eta)$ is however intractable. We use a sequential Monte Carlo algorithm~\citep{deFreitas2001,DelMoral2006} with target distribution $\Pr(d\theta_{(1:n)},d\tau_{(1:n)}|\Pi_n,\eta)$ in order to get unbiased estimators of the marginal likelihoods $\Pr(\Pi_n|\eta)$ for a grid of values of $\eta$, and compute the maximum likelihood estimate $\widehat \eta$. The proposal distribution for the arrival times $\tau_{(n)}$ is a truncated normal on $[\tau_{(n-1)},\infty)$, while the proposal for the cluster location of a new cluster is uniform on $[0,\tau_{(n)}]$.  We also use a sequential Monte Carlo algorithm in order to sample from the predictive $\Pr(\Pi_{n+m}|\Pi_n,\widehat \eta)$ using Proposition \ref{posterior}.

\section{Random partitions and random multigraphs}
\label{sec:graphs}
The non-exchangeable random partition model proposed can be used to derive models for random multigraphs, see~\citep{Bloem-Reddy2017}. Recall that $\Pi_n=(A_{n,1},\ldots,A_{n,K_n})$, where the blocks are sorted in order of appearance. For each $i=1,2,\ldots$, let $c_i$ be the index of the cluster to which item $i$ belongs, that is $i\in A_{n,c_i}$ for all $n\geq i$. An undirected multigraph $G=\Phi(\Pi)$, possibly with self-loops and  with a countably infinite number of edges, is derived from the random partition $\Pi$ by
$$
G=((c_1,c_2),(c_3,c_4),\ldots)
$$
where each pair $(c_{2n-1},c_{2n})$ represents an undirected edge between the vertex $c_{2n-1}$ and the vertex $c_{2n}$. The set of vertices is either $\{1,\ldots,K\}$ if the partition has a finite number of blocks, or the set $\mathbb N$. Let $G_n$ be the restriction of $G$ to the first $n$ edges that is, to the first $2n$ items of $\Pi$. Then $K_{2n}$, the number of clusters in $\Pi_{2n}$, is also the number of vertices of $G_n$, $m_{2n,j}$ is the degree of vertex $j$, $j=1,\ldots,K_{2n}$ and $K_{2n,j}/K_{2n}$ is the proportion of vertices of degree $j$.

The multigraphs $G$ obtained by transformation of an exchangeable random partition form a subclass of the edge-exchangeable graphs~\citep{Crane2017,Cai2016}. This subclass is called rank one edge-exchangeable graphs by Janson~\citep{Janson2017a}. Of particular interest is the so-called Hollywood model~\citep{Crane2017}, obtained from a two-parameter CRP random partition. In this case, inherited from the properties of the associated random partition~\citep{Pitman2002}, one can obtain sparse multigraphs with power-law degree distribution. A consequence of the exchangeability assumption is the fact that the degree sequence grows linearly with the number of edges: for any vertex $j$, its degree $m_{2n,j}\asymp n$ almost surely as the number of edges $n$ tends to infinity. As shown in the following corollary of the results of Section~\ref{sec:properties}, our construction allows to obtain sparse multigraphs with power-law degree distribution and sublinear growth rate for degree sequences.

\begin{corollary}
Let $\Pi$ be a non-exchangeable partition with parameters $\alpha$ and $\rho$ verifying assumptions (A1-A3). Let $G=\Phi(\Pi)$ the associated random multigraph. For a subgraph $G_n$ corresponding to the first $n$ edges, let $K_{2n}$ be the number of vertices, $m_{2n,j}$ the degree of vertex $j$ and $K_{2n,r}$ the number of vertices of degree $r \geq 1$. Then, almost surely as the number of edges $n$ tends to infinity
\begin{align*}
m_{2n,j}&\asymp L_{\xi+1}^*(n)\,n^{1/(1+\xi)},~~j\geq 1\\
K_{2n}&\sim \widetilde \ell (n) (2n)^{(\sigma+\xi)/(1+\sigma)}\\
  \frac{K_{2n,r}}{K_{2n}}&\rightarrow\frac{\sigma\Gamma(r-\sigma)}{r!\Gamma(1-\sigma)},~~ r\geq 1
\end{align*}
where the slowly varying functions $L_{\xi+1}^*$ and $\widetilde \ell$ are defined in Equations \eqref{eq:Lstar} and \eqref{eq:elltilde}.

\end{corollary}

\section{Discussion}
\label{sec:discussion}

To obtain random partitions with the microclustering property, one option is to give up the exchangeability assumption, as we did in this paper. An alternative approach is to drop the Kolmogorov consistency assumption discussed in the introduction. Miller et al.~\cite{Miller2015} and Betancourt et al.~\cite{Betancourt2016}, who derived random partition models with the microclustering property, take this option, and consider a collection $(\Pi_n)_{n\geq 1}$ of finitely exchangeable random partitions of $[n]$ that do not define a (Kolmogorov-consistent) random partition of $\mathbb N$. Another related contribution is the work of ~\citep{zhou2016frequency} where the authors also define a collection $(\Pi_n)_{n\geq 1}$ of finitely exchangeable random partitions of $[n]$ that do not satisfy Kolmogorov consistency property; the authors emphasize that it is indeed a desirable feature for modeling frequencies of frequencies, which motivates their work. Their model is also based on some Poissonization idea.

There has been a lot of interest over the past years in the development of non-exchangeable partitions based on dependent Dirichlet processes and more generally dependent random measures~\citep{MacEachern1999,Griffin2006,Caron2007,Foti2015,Blei2011,Caron2016}. The focus of these works is rather different though, as they do not aim to capture/characterize the microclustering property. The model presented here builds on a Poisson construction on an augmented space, and is therefore somewhat reminiscent of the work of~\citep{Rao2009,Lin2010,Chen2013,donnelly1991}.

Bloem-Reddy and Orbanz \citep{Bloem-Reddy2017} considered a general class of exchangeable and non-exchangeable random partitions of $\mathbb N$, motivated by preferential attachments models for random multigraphs. For certain values of the parameters, it can generate partitions with the microclustering property, but with a somewhat different asymptotic behavior for the number of clusters. The microclustering property is obtained whenever the number of clusters grows linearly with the dataset \citep[Theorem 7]{Bloem-Reddy2017}. In our approach, the number of clusters always grows sublinearly, and the rate can be controlled by the properties of the L\'evy measure $\rho$.

\section{Experiments} \label{experiments}
In what follows we compare our non-exchangeable model to the two-parameter Chinese restaurant process~\citep{Pitman2002}. For the non-exchangeable model, we consider a GGP with mean measure $\rho(d\omega)\alpha(d\theta)=1/\Gamma(1-\sigma)\omega^{-1-\sigma}e^{-\zeta\omega}d\omega\xi \theta^{\xi-1}d\theta$ where the parameters $\xi\in \{1,2,3\}$, $\sigma\in[0,1)$ and $\zeta>0$ are unknown. For the two-parameter CRP, the two parameters $\sigma_2\in[0,1)$ and $\kappa_2>0$ are considered unknown.  Observed data are partitions of size $n$, partitioned into a training set of size $n_\text{train}$ and a test set of size $n_{\text{test}}$ where $n_\text{train}+n_{\text{test}}=n$. The parameters of each model are estimated on the training data using maximum likelihood with a grid of values for the parameters: 25 equidistant points in $[0,1)$ for $\sigma$, $\xi\in \{1,2,3\}$, and a grid obtained by dichotomic search on the interval $[0,100]$ for $\zeta$, and similarly for the two-parameter CRP. The EPPF $\Pr(\Pi_{n_\text{train}}|\sigma_2,\kappa_2)$ of the two-parameter CRP has an analytic form and is calculated directly. For our method, we approximate the likelihood $\Pr(\Pi_{n_\text{train}}|\xi,\sigma,\zeta)$ using sequential Monte Carlo methods with 10000 particles, as described in Section \ref{inference}.

\begin{figure}
\begin{center}
\includegraphics[width=10cm]{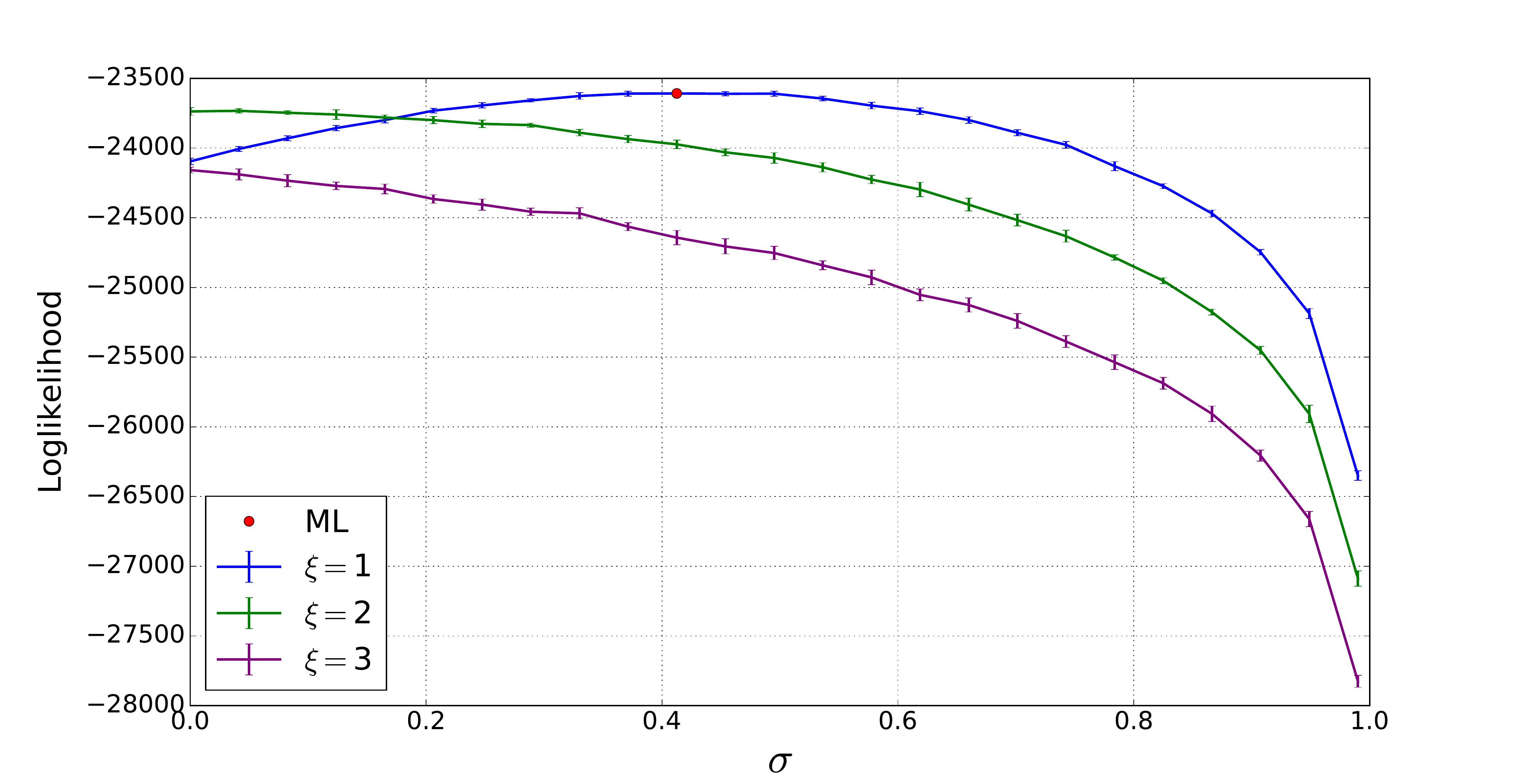}
\end{center}
\caption{Loglikelihood estimates for $\zeta = 10$ and different values of $\xi$ and $\sigma$. For every grid point, 10 SMC estimates are obtained, and the mean and $\pm 1$ standard deviation error bars are reported.}
\label{fig:loglikelihood}
\end{figure}

For each cluster $j=1,\ldots,K_{n_\text{train}}$ in the training set, we then aim at predicting its size $m_{k,j}$ for $k=n_\text{train}+1,\ldots,n$. Let $m^{\text{(true)}}_{k,j}$ be the true size of cluster $j$ in the partition of size $k$ and consider the L2 error
$$
E=\frac{1}{K_{n_{\text{train}}}}\sum_{j = 1}^{K_{n_{\text{train}}}}\frac{1}{n_\text{test}} \sum_{k=n_{\text{train}}+1}^{n}\left(m_{k,j} - m^{\text{(true)}}_{k,j}\right)^2\geq 0.
$$
We are interested in the distribution of the predictive error
\begin{equation}
\Pr\left(E\in dE \mid \Pi_{n_\text{train}},\widehat \eta\right )\label{eq:predictiveerrror}
\end{equation}
where $\widehat\eta$ are the fitted parameters, under the two-parameter CRP or our model.

Additionally, we want to check that the model can still capture the distribution of the cluster sizes adequately. To this aim, we also report $95\%$ predictive credible intervals for the proportion of clusters of a given size in the test set, and compare this to the empirical distribution.

\paragraph{Synthetic data.} In order to validate the inference procedure, we first run experiments on a simulated dataset, where the data are simulated from our model with parameters set to $(\xi,\sigma,\zeta)=(1,0.4125,10)$. In this model, the cluster size grows at a rate of $\sqrt n$, as can be seen from Figure~\ref{clustsize_plot}(a) that shows the growth of the cluster sizes with respect to the sample size $n$. Additionally, the proportion of clusters of a given size has an asymptotic power-law distribution, see the top row of Figure~\ref{fig:powerlawpred}. As shown in Figure~\ref{fig:loglikelihood}, the SMC estimate of the log-likelihood is rather accurate, and we recover the true parameters. The mean and quantiles of the predictive error under our model and the two-parameter CRP are reported in Table~\ref{tab:error}. As expected, the predictive under our model outperforms the two-parameter CRP, which is misspecified in that case. Posterior predictive of the proportion of clusters of a given size is reported in the first row of Figure~\ref{fig:powerlawpred}.

  \paragraph{Real data.}
We consider two datasets of the same size. The first one is the Amazon dataset of movies' reviews \cite{mcauley2015image} where each movie represents a cluster containing its reviews, which are ordered. The second dataset is a time-ordered collection of answers to questions in the Math Overflow website\footnote{https://mathoverflow.net/} where the clusters contain answers to the same question. Evolutions of the cluster sizes are reported in Figure~\ref{clustsize_plot}(b-c) for these datasets. We aim at predicting, based on the training set, the number of reviews to a given movie for the Amazon dataset, and the number of questions answered to a given question for the Math Overflow dataset. In both cases our non-exchangeable model provides better predictions of the cluster sizes (see Table~\ref{tab:error} and Figure~\ref{fig:amazonpredictive}). Estimates and credible intervals for the parameter $\sigma$ and $\sigma_2$ are reported in Table~\ref{tab:sigma}. Figure~\ref{fig:powerlawpred} shows that both models give reasonable predictive fit to the proportion of clusters of a given size.

\begin{figure}
\centering
\subfigure[Synthetic]
{\includegraphics[width=.31\textwidth]{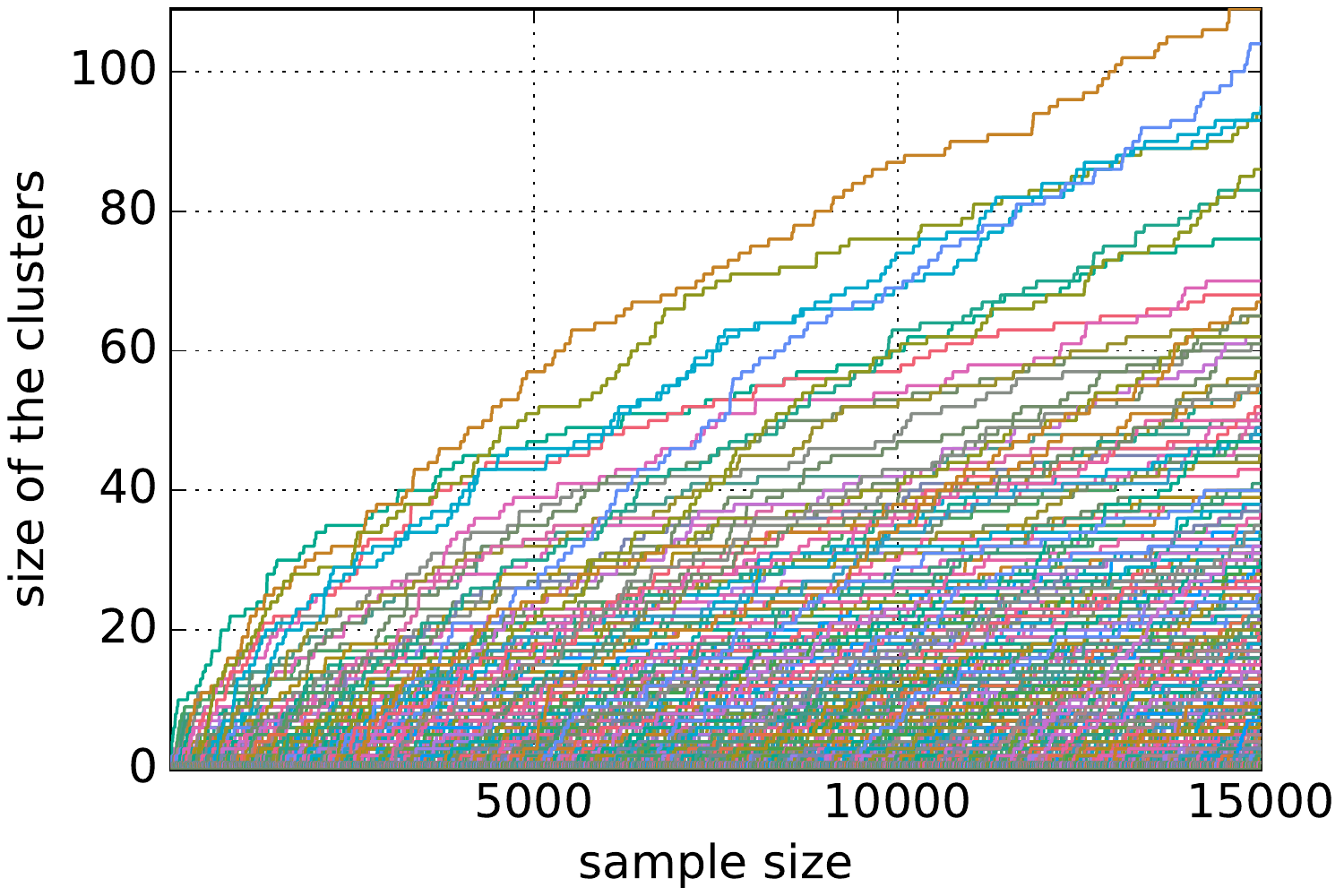}} \quad
\subfigure[Amazon]
{\includegraphics[width=.31\textwidth]{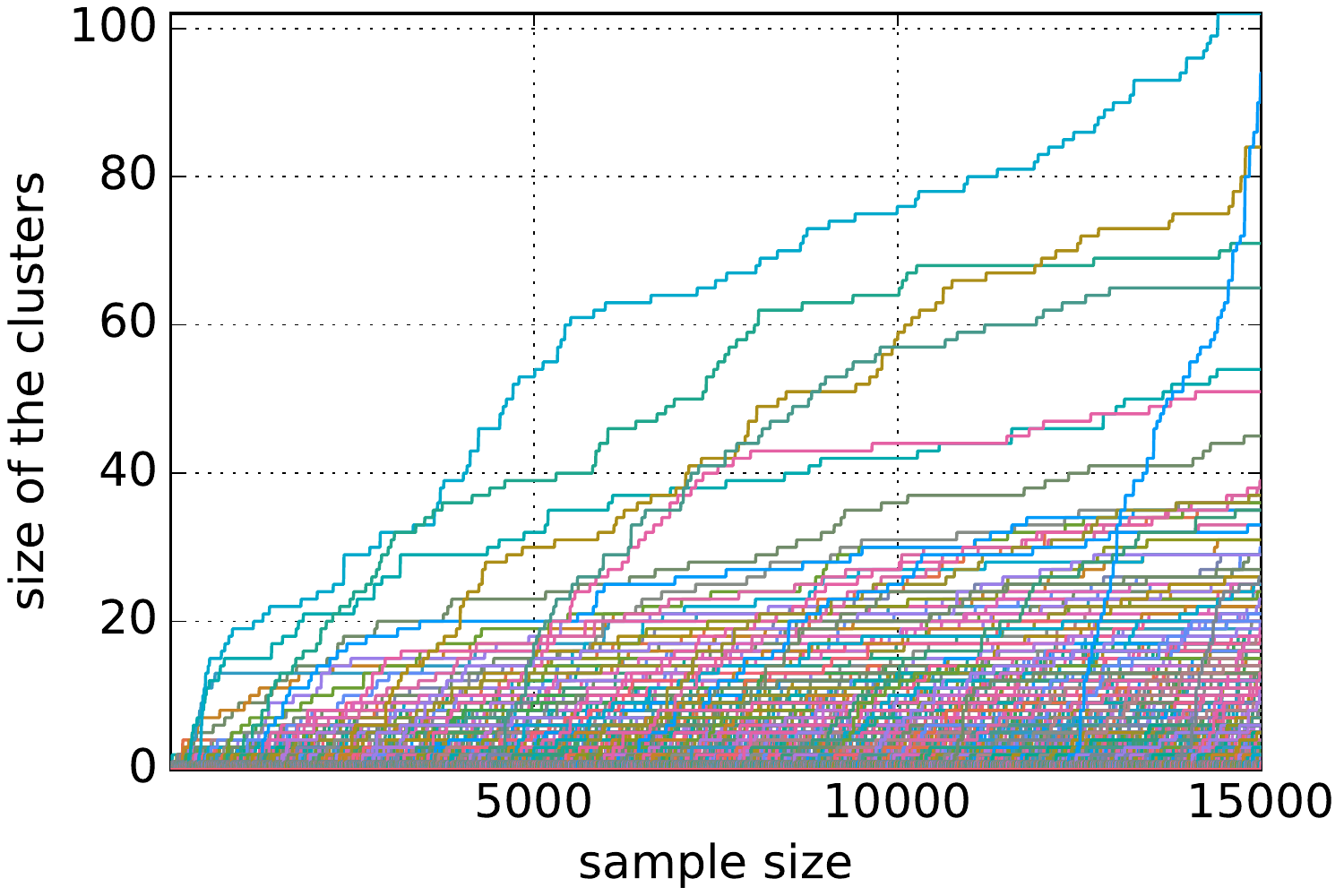}} \quad
\subfigure[Math Overflow]
{\includegraphics[width=.31\textwidth]{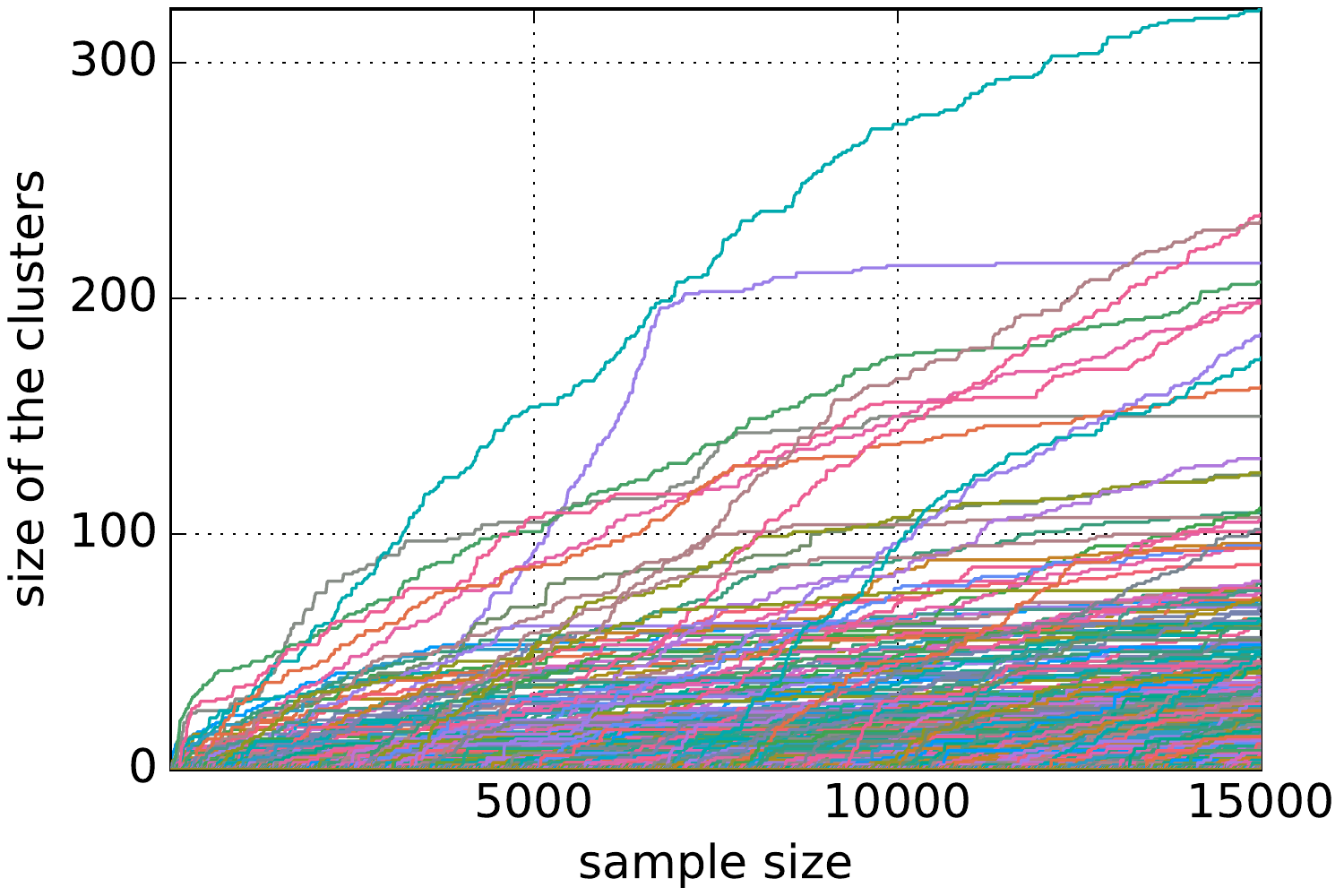}}
\caption{Evolution of the clusters' sizes $m_{j,n}$ with respect to the sample size $n$ for the (a) synthetic, (b) Amazon and (c) Math Overflow datasets.}
\label{clustsize_plot}
\end{figure}

\begin{table}[h]
\caption{Mean and quantiles of the predictive error using 100 samples from the predictive distributions.}
\label{tab:error}
\begin{center}
\begin{tabular}{rllll}
  \toprule
   & \multicolumn{2}{c}{Non-exchangeable} & \multicolumn{2}{c}{Two-parameter CRP} \\
   & \multicolumn{1}{c}{L2 error} & \multicolumn{1}{c}{$90\%$ CI} & \multicolumn{1}{c}{L2 error} & \multicolumn{1}{c}{$90\%$ CI} \\
  \midrule
  Synthetic & $6.92$ & $[6.37,\, 7.42]$ & $14.9$ & $[13.4,\,16.2]$\\
  Amazon & $4.14$ & $[3.91,\,4.40]$ & $6.05$  & $[5.63,\,6.49]$\\
  Math Overflow & $1.08\times 10^2$ & $[1.01,\,1.22]\times 10^2$ & $1.67\times 10^2$ & $[1.58,\, 1.77]\times 10^2$\\
  \bottomrule
\end{tabular}
\end{center}
\end{table}

\begin{table}
\caption{MLE for the the parameter $\sigma$ of the non-exchangeable model and the discount parameter $\sigma_2$ of the two-parameter CRP model, and 0.025 and 0.975 quantiles of their posterior distributions.}
\label{tab:sigma}
\begin{center}
\begin{tabular}{rclcl}
  \toprule
   & \multicolumn{2}{c}{Non-exchangeable} & \multicolumn{2}{c}{Two-parameter CRP} \\
   & \multicolumn{1}{c}{MLE of $\sigma$} & \multicolumn{1}{c}{$[q_{0.025},q_{0.975}]$} & \multicolumn{1}{c}{MLE of $\sigma_2$} & \multicolumn{1}{c}{$[q_{0.025},q_{0.975}]$} \\
  \midrule
  Synthetic & $0.413$ & $[0.381,\,0.445]$ & $0.463$ & $[0.414,\,0.505]$\\
  Amazon  & $0.578$ & $[0.523,\,0.609]$ & $0.434$ & $[0.373,\,0.484]$\\
  Math Overflow  & $0.371$ & $[0.339,\,0.403]$ & $0.304$ & $[0.238,\,0.360]$\\
  \bottomrule
\end{tabular}
\end{center}

\end{table}

\begin{figure}[h!]
\centering
\label{pred_clustsize}
\includegraphics[scale=.4]{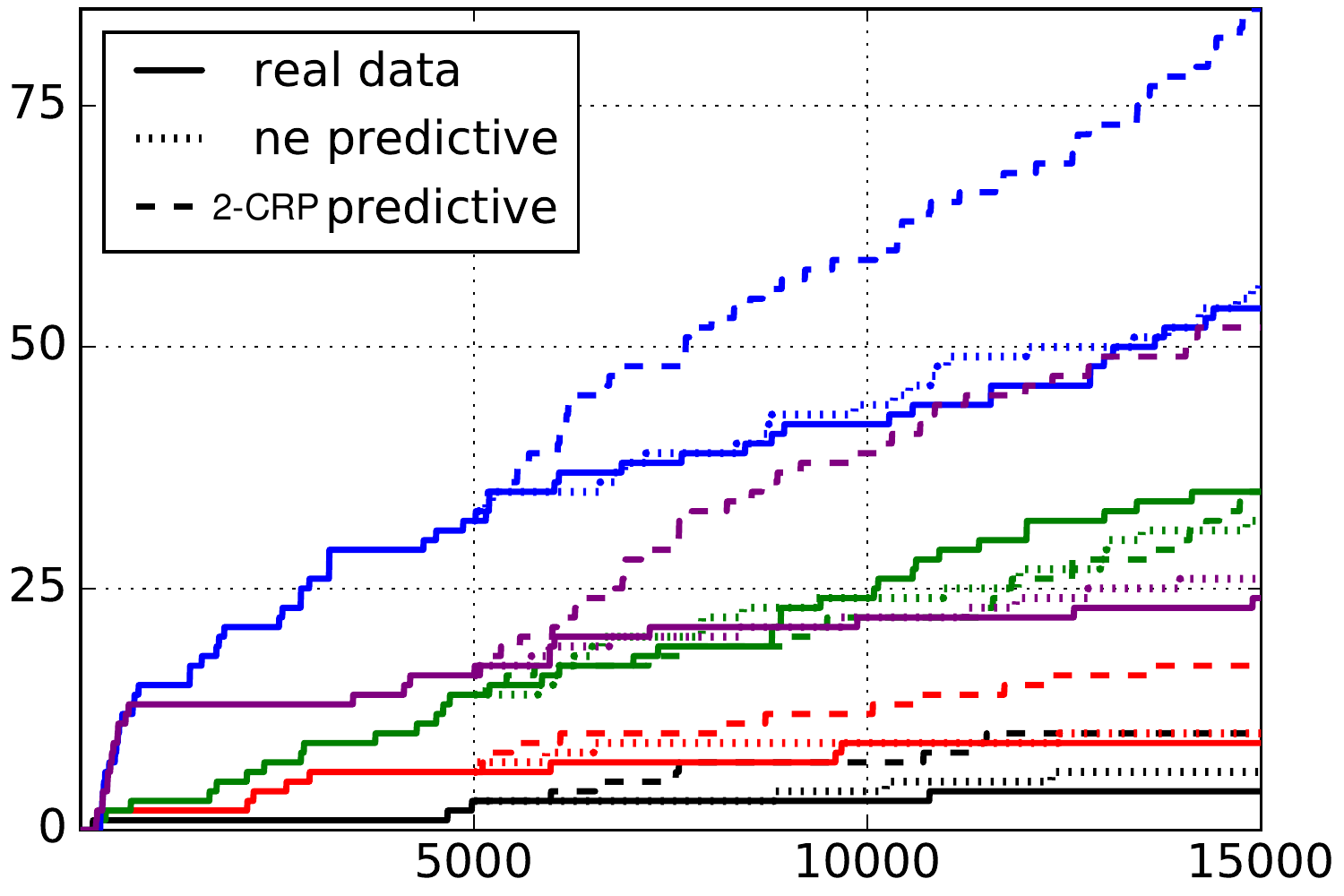}
\caption{Amazon dataset. Observed (plain line) and predicted sizes of some clusters (in different colours) from  the non-exchangeable (dotted line) and the two-parameter CRP models (dashed line). }
\label{fig:amazonpredictive}
\end{figure}


\begin{figure}[h!]
\centering
\subfigure[Synthetic]
{\includegraphics[width=.4\textwidth]{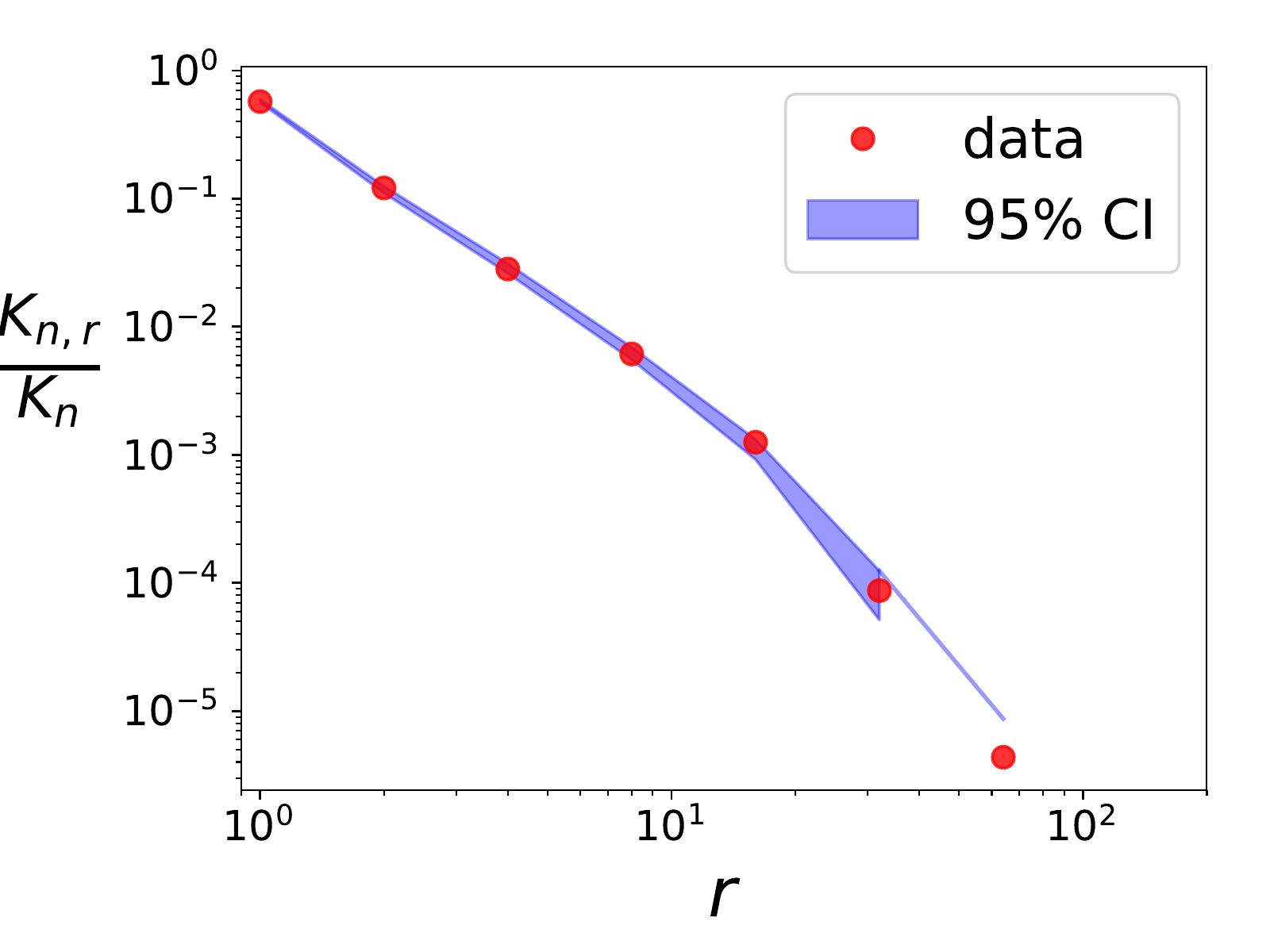} \quad \includegraphics[width=.4\textwidth]{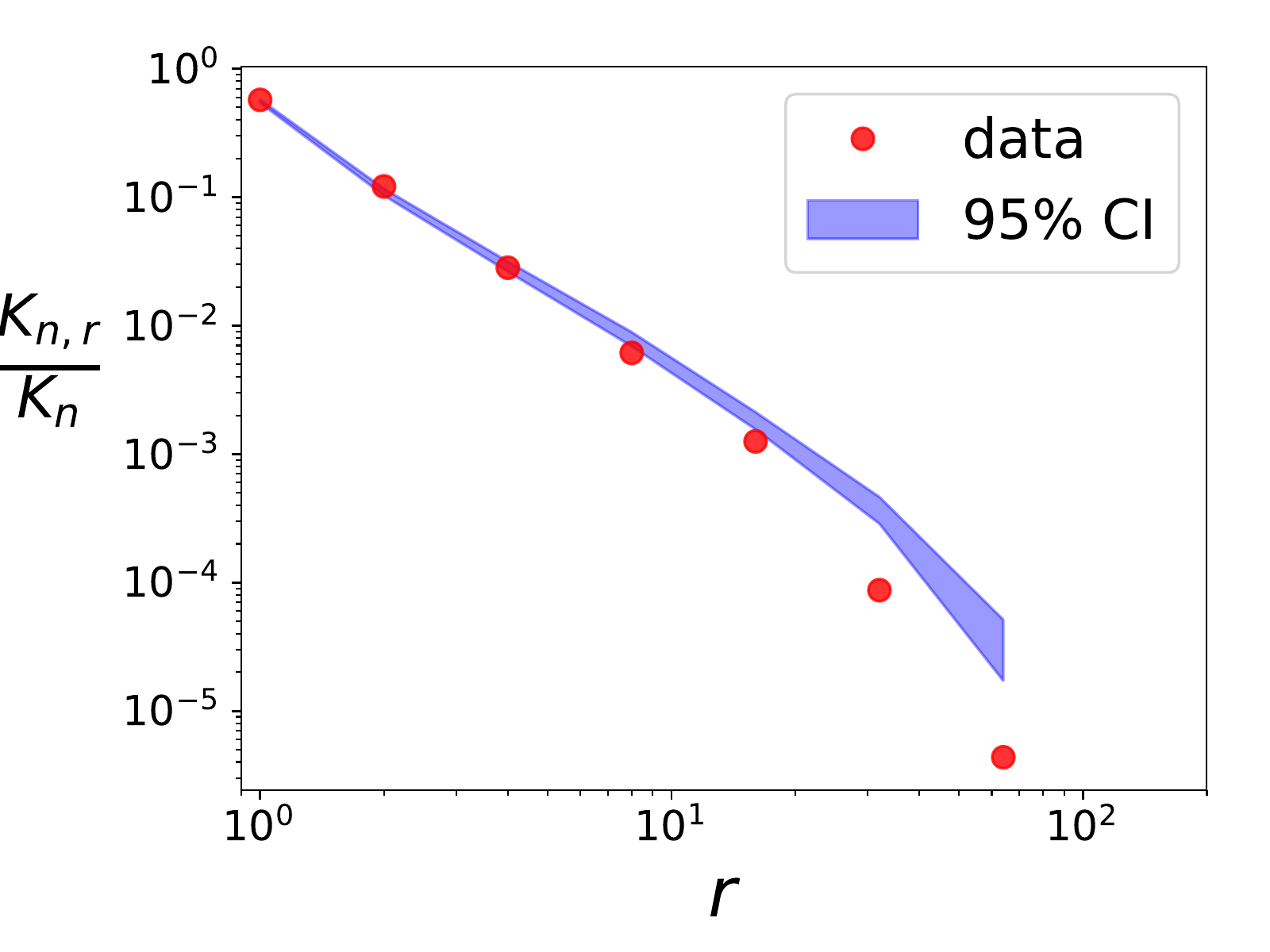}} \\
\subfigure[Amazon]
{\includegraphics[width=.4\textwidth]{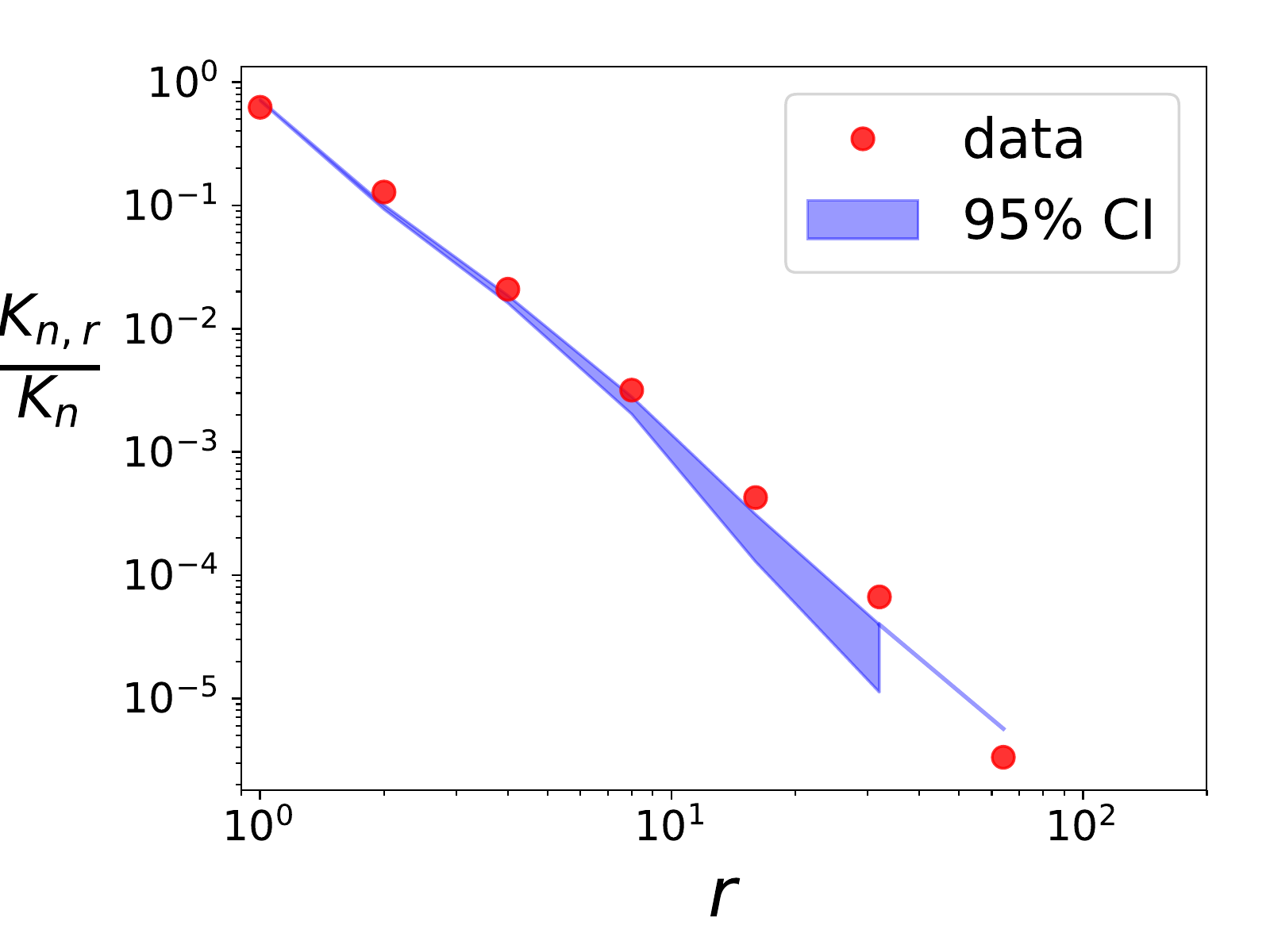} \quad \includegraphics[width=.4\textwidth]{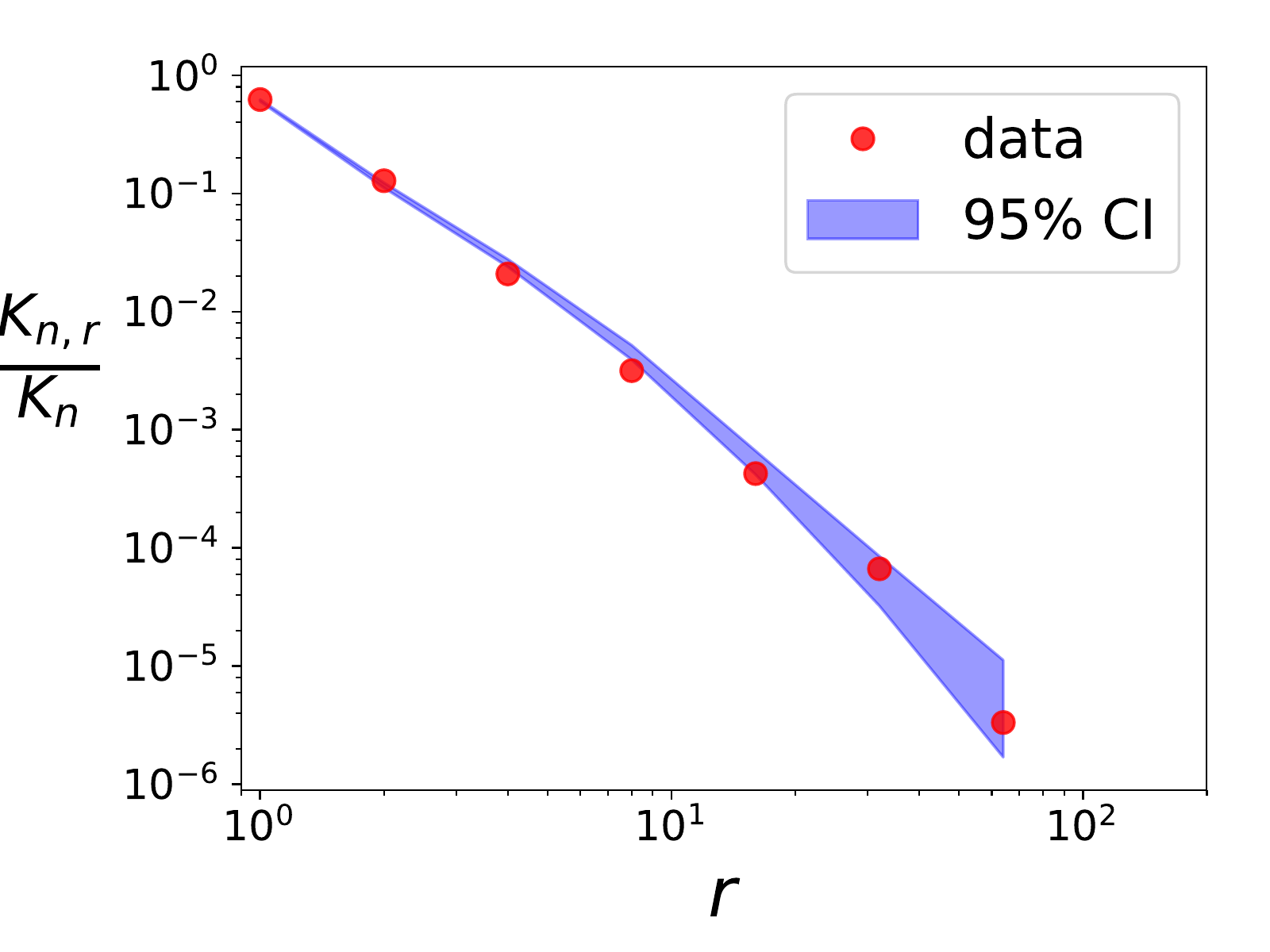}} \\
\subfigure[Math Overflow]
{\includegraphics[width=.4\textwidth]{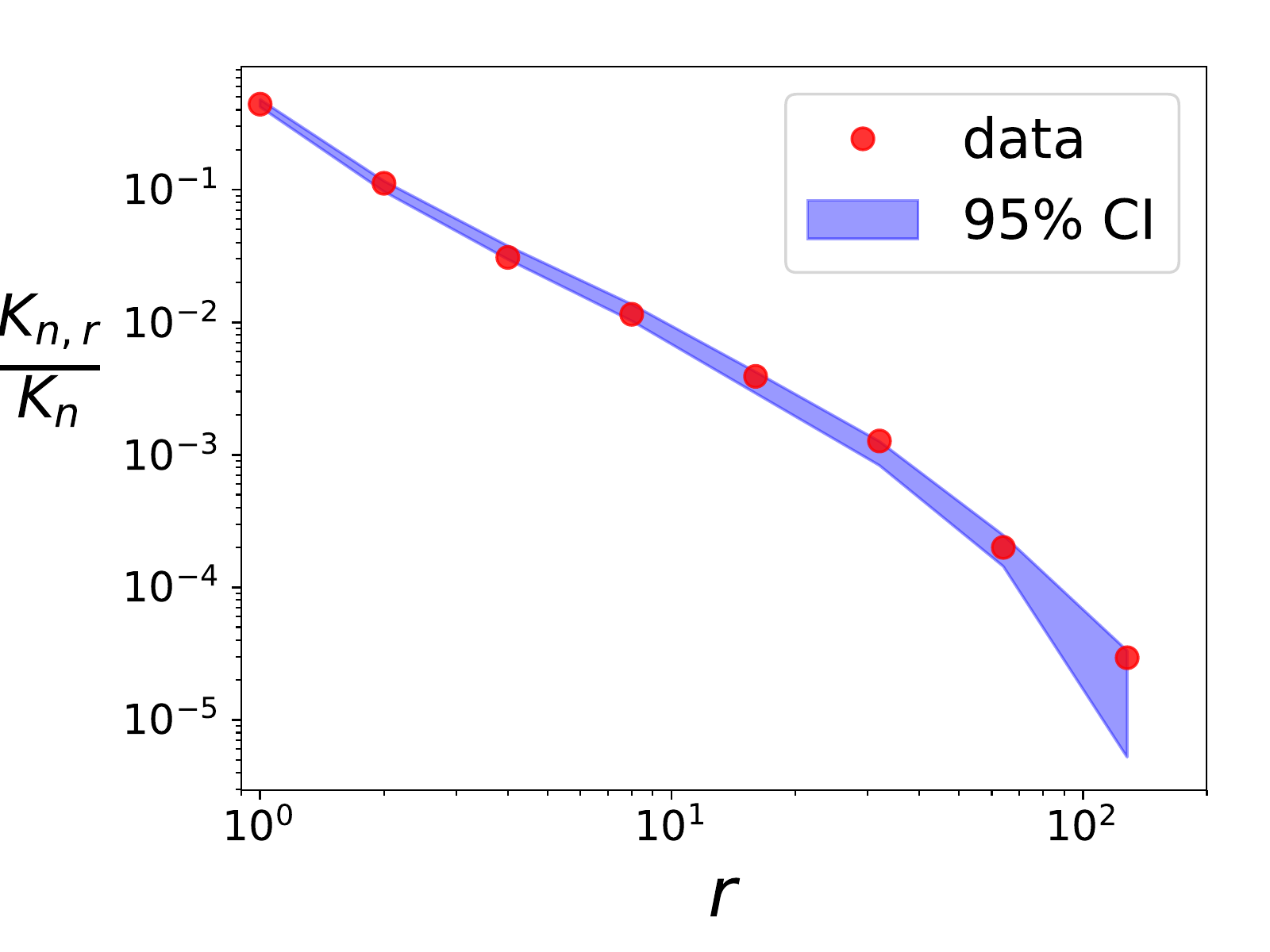} \quad \includegraphics[width=.4\textwidth]{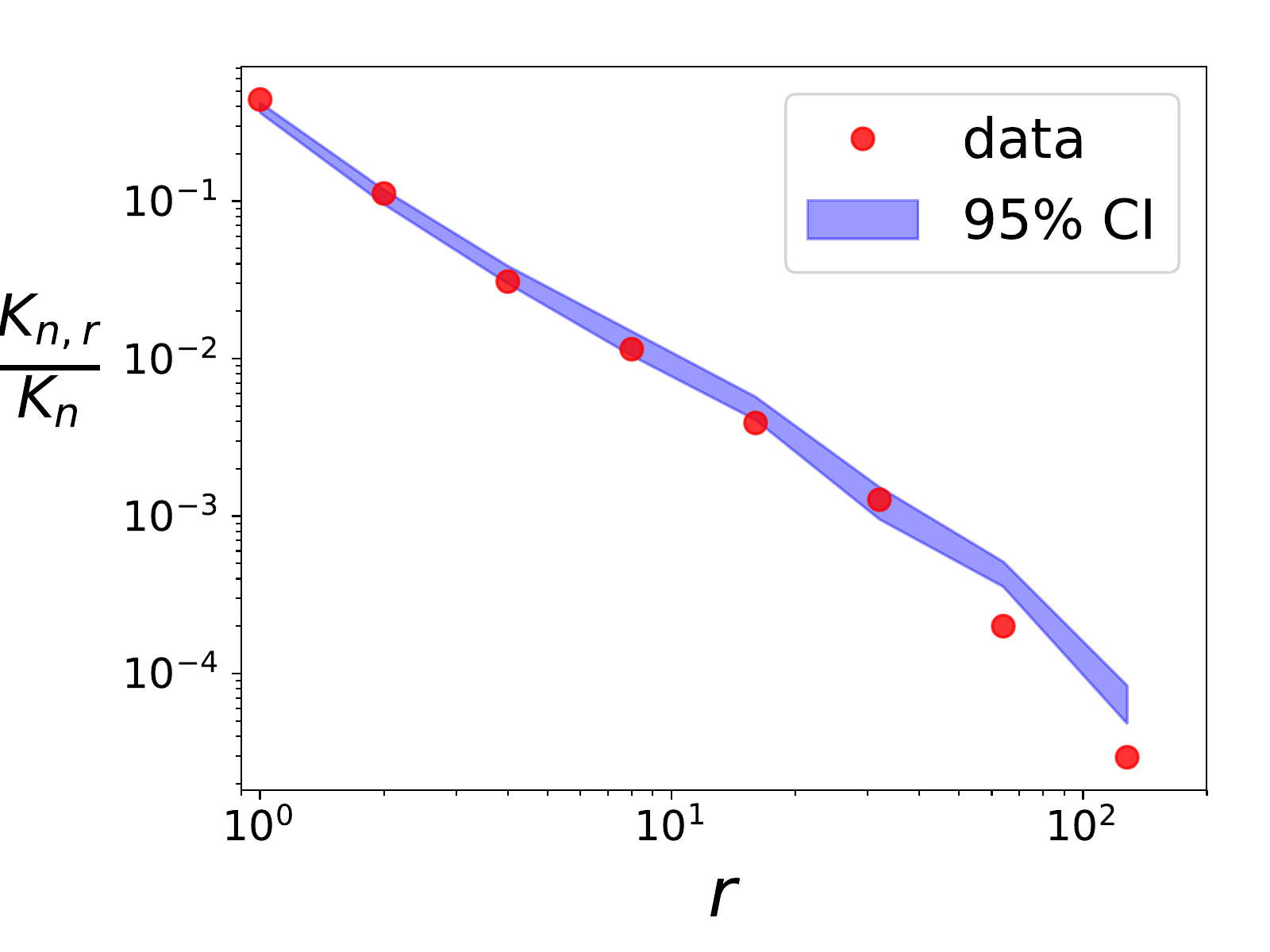}}
\caption{Empirical proportions of clusters of given size (red dots) and  $95\%$  posterior predictive credible intervals (blue) for our non-exchangeable model (left) and the two-parameter CRP (right).}
\label{fig:powerlawpred}
\end{figure}

\section*{Supplementary material}
A demo of the simulation and inference for the non-exchangeable random partition model can be found at \url{https://github.com/giuseppedib/microclustering}.

\section*{Acknowledgements}
Giuseppe Di Benedetto is funded by EPSRC (grant reference code EP/L016710/1).

\newpage

\appendix

\section{Proof of the main theorems}

\begin{proof}[Proof of proposition ~\ref{predictive}]
From Eq.~\eqref{eq:joint}, the joint distribution of
$((\theta_{(i)}, \tau_{(i)})_{i=1,\dots, n})$ is given by
\begin{align*}
  \Pr(d\theta_{(1:n)},& d\tau_{(1:n)}) =\Bigg\{ \sum_{i=1}^{K_{n-1}}
  \left[ \prod_{\substack{j=1\\ j\neq i}}^{K_{n-1}}
    \kappa(m_{n-1,j}, \tau_{(n)}-\theta_j^*)\alpha (\theta_j^*)\right]
 \kappa(m_{n-1,i}+1, \tau_{(n)}-\theta_i^*)\alpha (\theta_i^*)  \,
  \delta_{\theta_i^*}(d\theta_{(n)})\\
  &+ \left[\prod_{j=1}^{K_{n-1}}\kappa(m_{n-1,j}, \tau_{(n)}-\theta_j^*)\alpha (\theta_j^*)\right]
  \kappa(1, \tau_{(n)}-\theta_n^*) \alpha(\theta_{n}^*)\,
  d\theta_{(n)} \Bigg\}\,\\
  &\times e^{-\int_0^{\tau_{(n)}} \psi(\tau_{(n)}-\theta)\alpha(d\theta)}
\left[\prod_{i=1}^{n}\1{\theta_{(i)}<\tau_{(i)}}\right]
  \1{\tau_{(1)}<\tau_{(2)}<\ldots<\tau_{(n)}} d\theta_{(1:n-1)}d\tau_{(1:n)}.\end{align*}
Integrating over $\theta_{(n)}$, we obtain
\begin{align*}
\Pr(d\theta_{(1:n-1)},& d\tau_{(1:n)}) =
\Bigg\{ \sum_{i=1}^{K_{n-1}}\left[ \prod_{\substack{j=1\\ j\neq i}}^{K_{n-1}}
\kappa(m_{n-1,j}, \tau_{(n)}-\theta_j^*)\alpha (\theta_j^*)\right]\kappa(m_{n-1,i}+1, \tau_{(n)}-\theta_i^*)\alpha (\theta_i^*) \\
&+\left[\prod_{j=1}^{K_{n-1}}\kappa(m_{n-1,j}, \tau_{(n)}-\theta_j^*)
  \alpha (\theta_j^*)\right]\int_0^{\tau_{(n)}} \kappa(1, \tau_{(n)}-\theta_{(n)}) \alpha(\theta_{(n)})\,d\theta_{(n)}\Bigg\}\\
&e^{-\int_0^{\tau_{(n)}} \psi(\tau_{(n)}-\theta)\alpha(d\theta)} \left[\prod_{i=1}^{n-1}\1{\theta_{(i)}<\tau_{(i)}}\right]\1{\tau_{(1)}<\tau_{(2)}<\ldots<\tau_{(n)}} d\theta_{(1:n-1)}d\tau_{(1:n)}.
\end{align*}
In the Generalised Gamma Process case we have
\[
\kappa(m,u) = \frac{1}{\Gamma(1-\sigma)}
\frac{\Gamma(m-\sigma)}{(\zeta + u)^{m-\sigma}}.
\]
Therefore $\kappa(m+1,u) = \kappa(m,u) \frac{m-\sigma}{\zeta+u}$ and
\[
\sum_{i=1}^{K_{n-1}}\left[ \prod_{\substack{j=1\\ j\neq i}}^{K_{n-1}}
\kappa(m_{n-1,j}, \tau_{(n)}-\theta_j^*)\alpha (\theta_j^*)\right]
\kappa(m_{n-1,i}+1, \tau_{(n)}-\theta_i^*)\alpha (\theta_i^*)
\]\[= \left[ \prod_{j=1}^{K_{n-1}}\kappa(m_{n-1,j}, \tau_{(n)}-\theta_j^*)\alpha (\theta_j^*)\right]\sum_{i=1}^{K_{n-1}} \frac{m_{n-1,i}-\sigma}{\tau_{(n)}-\theta_i^*+\zeta}
\]
hence
\begin{align*}
\Pr(d\theta_{(1:n-1)},&d\tau_{(1:n)}) =
\frac{1}{\Gamma(1-\sigma) ^{K_{n-1}}}
\left[\prod_{j=1}^{K_{n-1}}\frac{\Gamma(m_{n-1,j}-\sigma)\,\alpha(\theta_j^*)}{(\tau_{(n)}-\theta_j^*+\zeta)^{m_{n-1,j}-\sigma}}\right]\\
&\times\Bigg(\sum_{j=1}^{K_{n-1}}\frac{m_{n-1,j}-\sigma}{\tau_{(n)}-\theta_j^*+\zeta} +\int_0^{\tau_{(n)}}\frac{\alpha(\theta)}{(\tau_{(n)} - \theta+\zeta)^{1-\sigma}}\, d\theta \Bigg)\\ &\times
e^{-\int_0^{\tau_{(n)}} \psi(\tau_{(n)}-\theta)\alpha(d\theta)}\left[\prod_{i=1}^{n-1}\1{\theta_{(i)}<\tau_{(i)}}\right] \1{\tau_{(1)}<\ldots<\tau_{(n)}}
d\theta_{(1:n-1)}d\tau_{(1:n)}
\end{align*}
from which we obtain the results of the theorem.
\end{proof}

\vspace*{0.5cm}

\begin{proof}[Proof of proposition~\ref{thm:clustsize}]
Given $W$, $N(t)$ is a non-homogeneous Poisson process with rate $\overline W(t)=\sum_{j\geq 1}\omega_j\1{\vartheta_j\leq t}$. Hence, using Fubini's and Campbell's theorems,
$$
\mathbb{E}[N(t)]=\mathbb{E}\left[\int_0^t \overline W(x)dx\right ]=\doverline{\alpha}(t)\kappa(1,0)$$ 
where $\doverline{\alpha}(t)=\int_0^t \overline \alpha(t)$.
Similarly,
 \begin{align*}
\var(N(t)) &= \var(\mathbb{E}[N(t)\,|\,W]) + \mathbb{E}[\var(N(t)\,|\,W)] = \var(\overline{\overline W}(t)) + \mathbb{E}[\overline{\overline W}(t)]\\
&= \kappa(2,0)\int_0^t(t-\theta)^2\alpha(d\theta) + \kappa(1,0)\doverline{\alpha}(t) \\
&= \kappa(2,0)\int_0^t \overline{\overline\alpha}(x)dx+ \kappa(1,0)\overline{\overline\alpha}(t)
\end{align*}
Using Karamata's theorem~\citep[Proposition 1.5.8]{Bingham1987} and Assumption (A3), we obtain that
$$\mathbb{E}[N(t)]\sim \frac{\kappa(0,1)}{\xi+1}t^{\xi +1}L(t)\text{ and }\var(N(t))\sim \frac{\kappa(2,0)}{(\xi+1)(\xi+2)}t^{\xi +2}L(t).$$ Therefore, for any $0<a<\xi$ we have $\var(N(t))=O(t^{-a}\mathbb E[N(t)^2])$. Using \citep[Lemma B.1]{Caron2017a}, we conclude that, almost surely as $t$ tends to infinity
\begin{align*}
N(t)&\sim\mathbb E[N(t)]\sim \frac{\kappa(1,0)}{\xi +1}t^{\xi+1}L(\xi).
\end{align*}

Conditional on $W$, $X_j(t)$ is a non-homogeneous Poisson process with rate $\omega_j \1{\vartheta_i>t}$ hence
$$
X_j(t)\sim \omega_j t
$$
almost surely as $t$ tends to infinity. It follows similarly that $M_j(t)$, the size, at time $t$, of the $j$th cluster to appear, satisfies
$$
M_j(t)\sim \omega_j^\ast\, t
$$
where $\omega_j^\ast=W(\{\theta_j^\ast\})$. Additionally, Propositions~\ref{likelihood} and \ref{posterior} imply that
\begin{align*}
\Pr(d\omega^\ast_1, d\theta_{(1)},d\tau_{(1)})&=\omega_1^\ast e^{-\omega_1^\ast (\tau_{(1)}-\theta_{(1)})}\rho(d\omega_1^\ast)\alpha(d\theta_{(1)})
e^{-\int_0^{\tau_{(1)}} \psi(\tau_{(1)}-\theta )\alpha(d\theta) } \1{\theta_{(1)}%
    <\tau_{(1)}}
\,d\tau_{(1)}.
\end{align*}
It follows
\begin{align*}
\Pr(d\omega^\ast_1)=\omega_1^\ast \rho(d\omega_1^\ast)\int_0^\infty\int_0^\tau e^{-\omega_1^\ast (\tau-\theta)}
e^{-\int_0^t \psi(\tau-u )\alpha(du) }
\alpha(d\theta)d\tau.
\end{align*}
\end{proof}

\begin{proof}[Proof of proposition ~\ref{thm:nbclust}]
Observe that $\Pr(X_j(t)>0\mid W)=1-e^{-\omega_j(t-\theta_j)_+}$. By the marking theorem~\citep[Chapter 5]{Kingman1993}, for each $t$, $\{(\omega_j,\vartheta_j)\mid j\geq 1,X_j(t)>0\}$ is a Poisson point process with mean measure $\rho(d\omega)\alpha(d\theta)(1-e^{-\omega(t-\theta)_+})$. It follows that\[\mathbb E[K(t)]=\var[K(t)]=\int_0^\infty\int_0^t \left(1-e^{-(t-\theta)\omega}\right)
\alpha(d\theta)\rho(d\omega)\\
=\int_0^t \psi(t-\theta)\alpha(d\theta).
\]
Similarly to \citep[Proposition 2]{Gnedin2007}, it follows from the monotonicity of $K(t)$ and the Borel-Cantelli lemma that $K(t)\sim \mathbb E[K(t)]$ almost surely as $t\rightarrow\infty$. Using the Tauberian theorems~\citep[Chapter XIII, Section 5]{Feller1971} recalled in Lemma~\ref{lemma:tauberian}, Lemma~\ref{lg} and $\alpha(t)\sim \xi t^{\xi-1}L(t)$, we obtain
$$\mathbb E[K(t)]\sim\frac{\Gamma(\sigma+1)\Gamma(\xi+1)}{\Gamma(\sigma+\xi+1)}\,L(t)\ell_\sigma(t)\, t^{\sigma+\xi}$$
as $t$ tends to infinity.

We proceed similarly for $K_r(t)$. For each $t> 0$ and $r\geq 1$, $\{(\omega_j,\vartheta_j)\mid j\geq 1,X_j(t)=r\}$ is a Poisson point process with mean measure $\rho(d\omega)\alpha(d\theta)\frac{\omega^r(t-\theta)^r_+}{r!} e^{-\omega(t-\theta)_+}$. It follows that
\begin{align*}
\mathbb E[K_r(t)]&=\var[K_r(t)]\\
&=\int_0^t\int_0^\infty \frac{(t-\theta)^r}{r!}\,
  \omega^re^{-(t-\theta)\omega} \rho(d\omega)\alpha(d\theta)\\
  &=\int_0^t\frac{(t-\theta)^r}{r!}\kappa(r,t-\theta)\alpha(d\theta).
\end{align*}
Using Lemma~\ref{lemma:tauberian} and~\ref{lg}, we obtain: if $\sigma=0$, $\mathbb E [K_r(t)]=o(L(t)\ell(t))t^{\sigma+\xi}$; if $\sigma\in(0,1)$,
 \begin{align*}
  \mathbb E [K_r(t)]\sim\frac{\sigma\Gamma(r-\sigma)}{r!\Gamma(1-\sigma)}\frac{\Gamma(\sigma+1)\Gamma(\xi+1)}{\Gamma(\sigma+\xi+1)}\,L(t)\ell_\sigma(t)\, t^{\sigma+\xi} .
  \end{align*}
    If $\sigma=1$, $\mathbb E[K_1(t)]\sim \frac{\Gamma(\sigma+1)\Gamma(\xi+1)}{\Gamma(\sigma+\xi+1)}\,L(t)\ell_\sigma(t)\, t^{\sigma+\xi} $ and $\mathbb E[K_r(t)]=o(L(t)\ell(t))t^{\sigma+\xi}$ for all $r\geq 2$. For the almost sure result, we proceed as for $K(t)$~\citep{Gnedin2007}, using the monotonicity of $\sum_{r\geq s} K_r(t)$, the equality $\var\left[\sum_{r\ge s}K_r(t)\right] = \mathbb{E}\left[\sum_{r\ge s} K_r(t)\right]$ and the fact that $\mathbb E[K_r(t)]\asymp K(t)$ for $\sigma\in(0,1)$, $\mathbb E[K_r(t)]=o(K(t))$ for $\sigma=0$ and $\mathbb E[K_1(t)]\sim K(t)$ for $\sigma=1$.
\end{proof}

\section{Background on regular variation and technical lemma}\label{appendix_regvar}

We recall the following definitions which can be found in \cite{Bingham1987} and \cite{resnick2007heavy}.
\begin{definition}[Regularly varying function]
  A measurable function $f:\mathbb{R}_+\rightarrow\mathbb{R}_+$ is \emph{regularly varying} at $\infty$ with index $\xi\in\mathbb{R}$ if for every $x>0$
  \[
  \lim_{t\rightarrow\infty}\frac{f(tx)}{f(t)} = x^\xi.
  \]
\end{definition}
If $\xi=0$ we say that the function is \emph{slowly varying}. An important property of the regularly varying function is that they can be written as $f(x) = \ell(x) x^\xi$ where $\xi$ is the exponent of variation and $\ell$ is a slowly varying function.\medskip

Let $L^{\#}$ be the de Brujin conjugate ~\citep{Bingham1987} of a slowly varying function $L$.
Regularly varying functions $f(x) = L(x)x^\xi$ of index $\xi >0$ admit asymptotic inverse $g(x)=L_\xi ^*(x)x^{1/ \xi}$ which are regularly varying of index $\xi^{-1}$ (see ~\citep[Proposition 1.5.15]{Bingham1987} or \citep[Lemma 22]{Gnedin2007}) with slowly varying part
\begin{align}
L_\xi^*(x)=\{L^{1/\xi}(x^{1/\xi})\}^{\#}.\label{eq:Lstar}
\end{align}
Note that if $L(t)=c$, then $L_\xi^\ast(t)=c^{1/\xi}$. From Equation~\eqref{eq:inversiont} and Proposition~\ref{thm:nbclust}, it follows that the slowly varying function appearing in Corollary~\ref{corollarypowerlaw} is
\begin{align}
 \widetilde \ell(n)&=\frac{\Gamma(\sigma+1)\Gamma(\xi+1)}{\Gamma(\sigma+\xi+1)} \left (\frac{\xi +1}{\kappa(1,0)}\right )^{(\sigma+\xi)/(\xi+1)} L_{\xi+1}^\ast(n)^{\sigma+\xi}\nonumber\\
  &\quad \times L\left\{\left (\frac{\xi +1}{\kappa(1,0)}\right )^{1/(\xi+1)} n^{1/(\xi+1)}L_{\xi+1}^\ast(n)\right\}\ell_\sigma\left \{\left (\frac{\xi +1}{\kappa(1,0)}\right )^{1/(\xi+1)} n^{1/(\xi+1)}L_{\xi+1}^\ast(n)\right \}.\label{eq:elltilde}
 \end{align}

\medskip

The following lemma is a compilation of Tauberian results in Propositions 17, 18 and 19 in \cite{Gnedin2007}. See also \citep[Chapter XIII]{Feller1971}.
\begin{lemma}\label{lemma:tauberian}
Let $\rho$ be a L\'evy measure on $(0,\infty)$ with tail L\'evy intensity $\overline\rho(x)=\int_x^\infty \rho(d\omega)$. Assume $$\overline\rho(x)\sim x^{-\sigma}\ell(1/x)$$
as $x$ tends to 0, where $\sigma\in[0,1]$ and $\ell$ is a slowly varying function at infinity. For  any $\sigma\in[0,1)$,
$$
\psi(t)\sim \Gamma(1-\sigma)t^\sigma\ell(t)
$$
and for $r=1,2,\ldots$
$$
\left \{
\begin{array}{ll}
  \kappa(r,t)\sim t^{\sigma-r}\ell(t)\Gamma(r-\sigma) & \text{if }\sigma\in(0,1) \\
  \kappa(r,t)=o(t^{\sigma-r}\ell(t)) & \text{if }\sigma=0
\end{array}
\right .
$$
as $t$ tends to infinity. For $\sigma=1$,
\begin{align*}
\psi(t)&\sim t\ell_1(t)\\
\kappa(1,t)&\sim \ell_1(t)
\end{align*}
and
$$
\kappa(r,t)\sim t^{1-r}\ell(t)\Gamma(r-1)
$$
for all $r\geq 2$ as $t$ tends to infinity, where $\ell_1(t)=\int_t^\infty x^{-1}\ell(x)dx$.\end{lemma}

\begin{lemma}\label{lg}
Let $f$ and $g$ be locally bounded, regularly varying functions with $f(x)=\ell_f(x)x^a$ and $g(x)=\ell_g(x)x^b$ where $a,b>-1$ and $\ell_f,\ell_g$ are slowly varying functions.Then as $t$ tends to infinity
\begin{align*}
\int_0^t f(x)g(t-x)dx &\sim \frac{\Gamma(a+1)\Gamma(b+1)}{\Gamma(a+b+2)} tf(t)g(t)
\\&\sim \frac{\Gamma(a+1)\Gamma(b+1)}{\Gamma(a+b+2)}\, \ell_f(t)\ell_g(t)t^{a+b+1}.
\end{align*}
\end{lemma}
\begin{proof}
Let us split the integral in the following way
\[
\int_0^tf(x)g(t-x)dx=\int_0^{\frac{t}{2}}f(x)g(t-x)dx+
\int_0^{\frac{t}{2}}f(t-x)g(x)dx.
\]

Let $\delta\in (0,\min(a,b)+1)$. From Potter's Theorem~\citep[Theorem 1.5.6]{Bingham1987}, there is $X$ such that for all $t>2X$, $u\in[ X/t,1/2]$,
$$
\frac{f(tu)}{f(t)}\leq 2 u^{a-\delta},~~~\frac{g(t(1-u))}{g(t)}\leq 2 (1-u)^{b-\delta}.
$$
Take $t>2X$. We have
\[
\int_X^{\frac{t}{2}} \frac{f(x)g(t-x)}{tf(t)g(t)}\,dx=
\int_0^{\frac{1}{2}}\frac{f(ut)}{f(t)}\frac{g\{(1-u)t\}}{g(t)}\1{u\in[X/t,1/2]}\, du
\]
where the integrand function is bounded by $4u^{a-\delta}(1-u)^{b-\delta}$ which is integrable, hence we have convergence to $\int_0^{\frac{1}{2}}u^a(1-u)^b du\in(0,\infty)$ by the dominated convergence theorem. We proceed analogously for the second integral.
Since $\int_{0}^{1}u^a(1-u)^b du=\frac{\Gamma(a+1)\Gamma(b+1)}{\Gamma(a+b+2)}$, we have the result.
\end{proof}

\bibliographystyle{imsart-nameyear}
\bibliography{reference}

\end{document}